\documentclass[12pt]{article}
\usepackage[dvips]{graphicx}

\usepackage{color}
\usepackage[x11names,rgb]{xcolor}

\usepackage{tikz}
\usetikzlibrary{arrows,snakes,backgrounds}

\usepackage{amsfonts}
\usepackage{amsmath}
\usepackage{amssymb}
\usepackage{amsthm}
\usepackage{enumerate}
\usepackage[multiple]{footmisc}
\usepackage{authblk}

\usepackage[pdftex]{hyperref}
\usepackage{natbib}
\bibpunct{(}{)}{;}{a}{,}{,}
\hypersetup{colorlinks,%
            citecolor=blue,%
            filecolor=blue,%
            linkcolor=blue,%
            urlcolor=blue,%
            }
            
\usepackage{subcaption}
\usepackage{soul}
\usepackage[font={small},labelfont=bf]{caption} % this sets the font of all captions

\renewcommand{\vec}[1]{\mathbf{#1}}
\newtheorem{proposition}{{\bf \sc Proposition}}

\newtheorem{lemma}[proposition]{{\bf \sc Lemma}}

\newtheorem{corollary}[proposition]{{\bf \sc Corollary}}
\newtheorem{definition}{{\bf \sc Definition}}

\newtheorem{assumption}{{\bf \sc Assumption}}
\theoremstyle{remark} \newtheorem{example}{{\bf \sc Example}}
\def\eproof{\qed}

\newenvironment{new_am}
{\begin{color}{black} \ignorespaces} %set "black" when you want the submission ready manuscript.
{\end{color}}

\newenvironment{rev}
{\begin{color}{black} \ignorespaces} %set "black" when you want the submission ready manuscript.
{\end{color}}

\begin{document}

\title{Efficiency and Stability in a Process of Teams Formation\thanks{We thank the following for helpful comments: Andrea Galeotti, Sanjev Goyal, Matthew Jackson, Shachar Kariv, Brian Rogers, Fernando Vega Redondo, Simon Weidenholzer, Leeat Yariv and seminar participants at the ASSET 2012 Meeting in Cyprus, Berkeley, Caltech, Essex University, European University Institute, LUISS University in Rome, University of Siena and Stanford University. We also thank Matteo Chinazzi for his help with the data.
We acknowledge funding from the Italian Ministry of Education Progetti di Rilevante Interesse Nazionale"
(PRIN) grants 2015592CTH and 2017ELHNNJ.}}

\author[1]{Leonardo Boncinelli}
\author[2]{Alessio Muscillo}
\author[2,3]{Paolo Pin}

\affil[1]{\small Dipartimento di Scienze per l'Economia e l'Impresa, Universit\`a di Firenze, Italy
%email: \texttt{leonardo.boncinelli@unifi.it.}
}

\affil[2]{Dipartimento di Economia Politica e Statistica, Universit\`a di Siena, Italy %email: \texttt{alessio.muscillo2@unisi.it.}
}

\affil[3]{BIDSA, Universit\`a Bocconi, Italy %email: \texttt{alessio.muscillo2@unisi.it.}
}

% \author{Leonardo Boncinelli\thanks{Dipartimento di Scienze per l'Economia e l'Impresa, Universit\`a degli Studi di Firenze, Via delle Pandette 9, 50127 Firenze, Italy. %Tel.: +39 055 2759578, fax: +39 055 2759910, 
% email: \texttt{leonardo.boncinelli@unifi.it.}}
% %, Via Cosimo Ridolfi 10, 56124 Pisa, Italy. Tel.: +39 050 221 6219, fax: +39 050 221 6384, 
% \and Alessio Muscillo\thanks{%
% Dipartimento di Economia Politica e Statistica, Universit\`a degli Studi di Siena, Piazza San Francesco 7, 53100 Siena, Italy. email: \texttt{alessio.muscillo2@unisi.it.}}
% \and Paolo Pin\thanks{%
% Dipartimento di Economia Politica e Statistica, Universit\`a degli Studi di Siena, Piazza San Francesco 7, 53100 Siena, Italy and IGIER \& BIDSA, Bocconi University, Milan, Italy. %tel.: +39 0577 232692, fax: +39 0577 232661, 
% email: \texttt{paolo.pin@unisi.it.}}
%}
\date{\today%VERY PRELIMINARY VERSION   \\ \vspace{16pt} 
%March 2014
}
\maketitle

\begin{abstract}
\noindent
Motivated by data on coauthorships in scientific publications, we analyze a \emph{team formation process} that generalizes matching models and network formation models, allowing for overlapping teams of heterogeneous size. We apply different notions of stability: \emph{myopic team-wise stability}, which extends to our setup the concept of pair-wise stability, \emph{coalitional stability}, where agents are perfectly rational and able to coordinate, and \emph{stochastic stability}, where agents are myopic and errors occur with vanishing probability. We find that, in many cases, coalitional stability in no way refines myopic team-wise stability, while stochastically stable states are feasible states that maximize the overall number of activities performed by teams.
\end{abstract}

%\begin{abstract}
%\noindent
%We analyze a \emph{team formation process} that generalizes matching models and network formation models, allowing for overlapping teams of heterogeneous size. We define a weak concept of stability, called \emph{myopic team-wise stability}, which extends to our setup the concept of pair-wise stability used in network formation models. Then we refine it in two ways: (i) through \emph{stochastic stability}, where agents are still myopic and errors occurring with vanishing probability can dissolve or create teams, and (ii) through \emph{coalitional stability}, where agents are perfectly rational and able to coordinate.
%We find that the first concept has in many cases a much stronger predictive power than the second one. In particular, under some stark assumptions, coalitional stability in no way refines myopic team-wise stability. By contrast, stochastically stable states are, under a very general assumption, feasible states that maximize the overall number of activities performed by teams. 
%%At a first approximation, this turns out to be welfare maximizing in many environments.
%Finally, we apply our main result to the marriage problem, and we use the marriage theorem to obtain a characterization of stochastically stable matchings.
%\end{abstract}


\bigskip


%\newpage

\section{Introduction}\label{section:intro}

\begin{new_am}
In this paper we propose a theoretical analysis of the process of team formation to perform tasks with the aim to shed light on a broad variety of real-world activities. 
Our model generalizes matching models and network formation models.

We consider a finite set of agents, with heterogeneous constraints on time, who have the possibility of choosing between a finite set of \emph{teams} and allocate their time to perform tasks. 
Teams of different size are allowed and can be formed only if they satisfy exogenous constraints, called \emph{technology}.
For instance, agents may form a team only if they are neighbors in an exogenously given social network, or if they match complementary exogenous skills, or if they have common communication tools.
If time constraints are satisfied for each agent, a configuration of teams is \emph{feasible} and called a \emph{state}. 
Each state provides a specific payoff to each agent.
As we discuss in Section \ref{subsec:generalization}, this setup generalizes matching models and network formation models.
\end{new_am}

We find that, under the assumption of non-satiation with respect to teams for every agent, an extension of the simple notion of \emph{pair-wise stability} (\citealp{JacksonWolinsky96}) to this setup -- which we call \emph{myopic team-wise stability} -- does not have a strong predictive power on the stable states of the model; in particular, feasible states that are maximal with respect to set inclusion are myopically team-wise stable.
% \footnote{\begin{rev}We recall that set inclusion is a partial order (over the state space) so a maximal element is a subset that is not contained in any other subset of the state space. Throughout we will also sometimes refer to the state where the maximum number of teams/tasks are formed and performed, but the two notions may differ.\end{rev}}
Therefore, we compare two possible refinements.
The first, \begin{rev}following the same approach of \cite{JacksonWatts02}\end{rev}, is given by \emph{stochastic stability}: in the presence of extremely rare errors that can create or dissolve teams, and with agents that adapt myopically to the current state of the system, we find that the stochastically stable states are those that actually maximize the overall number of teams.
The second is a generalization of \emph{coalitional stability} for cooperative games (\citealp{Gillies59}): its predictive capability will turn out to be heavily dependent on the assumptions on payoff functions.
Moreover, when all projects are equivalent for every agent in terms of costs (resources employed) and benefits (payoffs earned), we find that the states that satisfy this form of coalitional stability are exactly the same as those that are myopically team-wise stable. 
Therefore, this latter refinement -- which is much more demanding in terms of agents' rationality -- proves to have no greater predictive power with respect to myopic team-wise stability in a stark but significant case.

The paper is structured as follows.
In Section \ref{subsec:literature} we consider its relation to the extant literature. 
In Section \ref{subsec:empirical_motivation} we motivate the contribution by considering an empirical analysis to which our model can be applied.
In Section \ref{section:model} we present all the aspects of the model, without any definition of stability.
In Section \ref{section:mts} we introduce and discuss the weak notion of myopic team-wise stability, which is then refined with the tools of stochastic stability in Section \ref{section:ss}, and with a concept of coalitional stability in Section \ref{section:cs}.
% {\color{red}In Section \ref{section:am} we combine our result on stochastic stability with the marriage theorem to provide a characterization of perfect matchings.}
Section \ref{section:conclusion} lists possible extensions of our study, and some additional discussion and results are in the Appendices.

\section{Literature} \label{subsec:literature}

In the real world, activities are often performed by people in \emph{teams}, so that the constraints each agent has to take into consideration in her decision depend on the choices, and hence the constraints, of others. 
For instance, if Alice wants to allocate a couple of hours on Saturday morning to playing tennis, but all her friends have already fully allocated time on Saturday morning to other activities, then Alice's desire to play tennis will remain unsatisfied. This simple example shows the existence of indirect externalities that must be taken into account in every individual decision when activities are performed in teams.
% , bringing an additional layer of difficulty in the knapsack problem presented above.
\begin{new_am}
The same happens for team formation and co-authorships in academic research. People work simultaneously on different projects, often participating to different teams which, in turn, may also work simultaneously on several projects. 
Each researcher allocates her time for each team and each project and, hence, every team that is formed may generate negative (or positive) direct externalities for non-members, due for instance to the reduction of effort that a researcher puts into each single project when she undertakes a new project (as further described in Section \ref{subsec:empirical_motivation} and modeled in Section \ref{subsec:publishing}).
Since size and composition of teams are important drivers of performance, this team formation process is not only studied in the literature about the academic profession (e.g., by \citealp{milojevic2014principles}) but also in the literature on R\&D and entrepreneurial activity related to co-foundation of firms \citep{BreschiLissoni01,StuartSorenson08,shah2019jewels}.
\end{new_am}

\begin{rev}
This paper starts with the pair-wise stability notion defined by \cite{JacksonWolinsky96} and extends it to team-wise stability, whereby tasks can be performed by groups of more than 2 individuals and of different sizes. 
To refine this myopic equilibrium concept with stochastic stability we use the same approach of \cite{JacksonWatts02} while extending their model, which is a special case of the one presented here.
Moreover, we contribute to this stream of literature by incorporating constraints to the capacity of agents of performing tasks and of coordinating with others in the spirit of  \cite{staudigl2014constrained} and \cite{baumann2021model}, although both works consider a pure non-cooperative framework in contrast with our approach which is more cooperative.
\end{rev}

As far as matching models are concerned, the most recent papers that study multi-matching environments \begin{new_am}with more specific results\end{new_am} are \citet{Pycia12} and \citet{pycia2019matching}, \begin{new_am}concentrating on the presence of externalities\end{new_am}, and \citet{hatfield2014many}, \begin{new_am}with the focus on the effects of agents' specific preferences\end{new_am}. 
%In the matching perspective, we provide a very general characterization of stochastically stable matchings under the assumptions of the marriage theorem (see \citealp{boma} for a modern exposition, although this mathematical result is already in \citealp{hall1935representatives}).
With respect to network formation models, we generalize pair-wise stability (see \citealp{JacksonWolinsky96}) and strong stability (\citealp{jackson2005strongly}) to a more general setting of resource-constrained team formation. 
The constraint imposed on our agents by a fixed time resource has been modeled in network formation models, e.g., by \citet{staudigl2014constrained}.
On the other hand, the constraints imposed by the technology can be related to many aspects introduced in the network formation literature: 
constraints may be due to homophily (see \citealp{CJP09}), because only similar agents may be able to form a team together;
or they may be related to an exogenous network of opportunities, because only linked agents have the opportunity to match (on this, we are aware only of \citealp{FMP10});
or they may be imposed by complementary exogenous skills that need to be matched together (see, e.g., \citealp{currarini2016simple}).

Stochastic stability (for which the references are discussed in Section \ref{section:ss}) has been applied to networks (first by \citealp{JacksonWatts02}). More recently \citet{Klaus20102218} used stochastic stability as a predictive tool for roommate markets.
In \citet{BonPin12}, best shot games played in exogenous networks are analyzed, and stochastically stable states are proven to be the states with the maximum number of contributing agents if the error structure is such that contributing agents are much more likely to be hit by perturbations.

The stochastic stability analyses carried on in this paper generalizes the results in \citet{BonPin18}, where an application to marriage markets is considered, distinguishing between the case of a link-error model, where mistakes directly hit links, and the case of an agent-error model, where mistakes hit agents' decisions and only indirectly links. In this setting, stochastic stability proves ineffective for refinement purposes in the link-error model -- where all maximal matchings are stochastically stable -- while it proves effective in the agent-error model -- where all and only  matchings \begin{rev}that maximize the number of pairs formed\end{rev} are stochastically stable.

% Our concept of coalitional stability stems from concepts of cooperative game theory, and particularly from the literature on coalition formation (see, e.g., \citealp{konishi2003coalition}, \citealp{gomes2005dynamic} and \citealp{hyndman2007coalition}) and clubs (see, e.g., \citealp{Pauly70} and \citealp{faias2017endogenous}).
% We provide further references and undertake an in-depth discussion on this in Section \ref{section:cs}.
Our concept of coalitional stability stems from concepts of cooperative game theory, and particularly from the literature on coalition formation (see, e.g., \citealp{konishi2003coalition}, \citealp{gomes2005dynamic}, \citealp{hyndman2007coalition}, \citealp{chalkiadakis2010cooperative}, \citealp{sawa2014coalitional}, \citealp{mauleon2018constitutions} and \citealp{sawa2019stochastic}) and clubs (see, e.g., \citealp{Pauly70} and \citealp{faias2017endogenous}).
We provide further references and undertake an in-depth discussion on this in Section \ref{section:cs}.

\begin{new_am}
\section{Empirical motivation}
\label{subsec:empirical_motivation}

We consider the \emph{American Physical Society} dataset (APS), which comprises publications spanning several decades in virtually all fields of physics and also contains information about their references and their authors.\footnote{The dataset consists of around 300,000 authors and 570,000 papers published from the beginning of the XX century until 2015. To avoid the scarcity of data before the 1950s, here we only focus on the most consistent part of this dataset, thus limiting our attention to papers published from 1960 and, consequently, to authors whose career started after 1960 as well.} 
This allows building a co-authorship network where a link between two authors is established if and only if they are both authors of a paper and, in addition, it allows the computation of the citations received by every paper (and, hence, by each author).

To analyze the data at an individual level, we focus on authors that have a lengthy and consistent career of at least 25 years, thus restricting the sample to around 24,000 authors.\footnote{The career length of an author is determined by the years passed from her first to her last publication (recorded in APS). Her cohort is the first year of publication. The subsample selected consists of authors whose cohort is from 1960 to 1990 and whose median career length is around 32 years. These authors are consistently present over time in the dataset, since the median author has a publication recorded every 2 years. Additional information in \ref{app:data}.}
% Our model starts from the assumption (\ref{ass:v0}) that individuals tend to increase the number of projects in which they participate. Indeed, t
% The analysis of individuals' careers shows that, on the one hand, every researcher tends to participate to a stable - or slightly increasing - number of projects together with an increasing number of collaborators (Figure \ref{fig:1} top and bottom-left).
This analysis shows that, on the one hand, throughout one's career every researcher tends to participate to a stable - or slightly increasing - number of projects together with an increasing number of collaborators (Figure \ref{fig:1} top and bottom-left).
This suggests that researchers behave as if they always gain by entering new projects, even if they are already participating to many of them (this is in line with the assumptions of our model and, specifically, with Assumption \ref{ass:v0} introduced in Subsection \ref{subsec:assumptions}).
% This suggests that researchers behave as if they always gained by entering new projects, even if they are already participating to many of them (and it is specifically captured in our model by Assumption \ref{ass:v0}).
% This is in line with our main assumptions, specifically with Assumption \ref{ass:v0}.

On the other hand, \begin{rev}as shown in Figure \ref{fig:1} bottom-right (and Figure \ref{fig:2} in \ref{app:data}),\end{rev} the trends of papers produced and citations received -- when normalized for the number of authors%\footnote{\begin{rev}A weighted paper is counted as $\frac{1}{\text{n authors}}.$\end{rev}} 
-- seem to suggest that researchers tend to dedicate less energy and time to these projects and also their quality -- as proxied by citations -- decreases.
% This is also captured by our model in terms of the negative externalities that being part of a project exerts on other projects in which an author may be involved.
This is also captured by our model in terms of negative externalities: when an agent takes part in a project, this imposes a negative externality on agents that are her collaborators in other projects. 
In the application outlined in Section \ref{subsec:publishing}, researchers face the tension between the benefits accrued when authoring multiple projects with multiple coauthors and the communication and coordination costs that increase with team size.

\begin{figure}[ht]
    \centering
     \includegraphics[width=.75\textwidth]{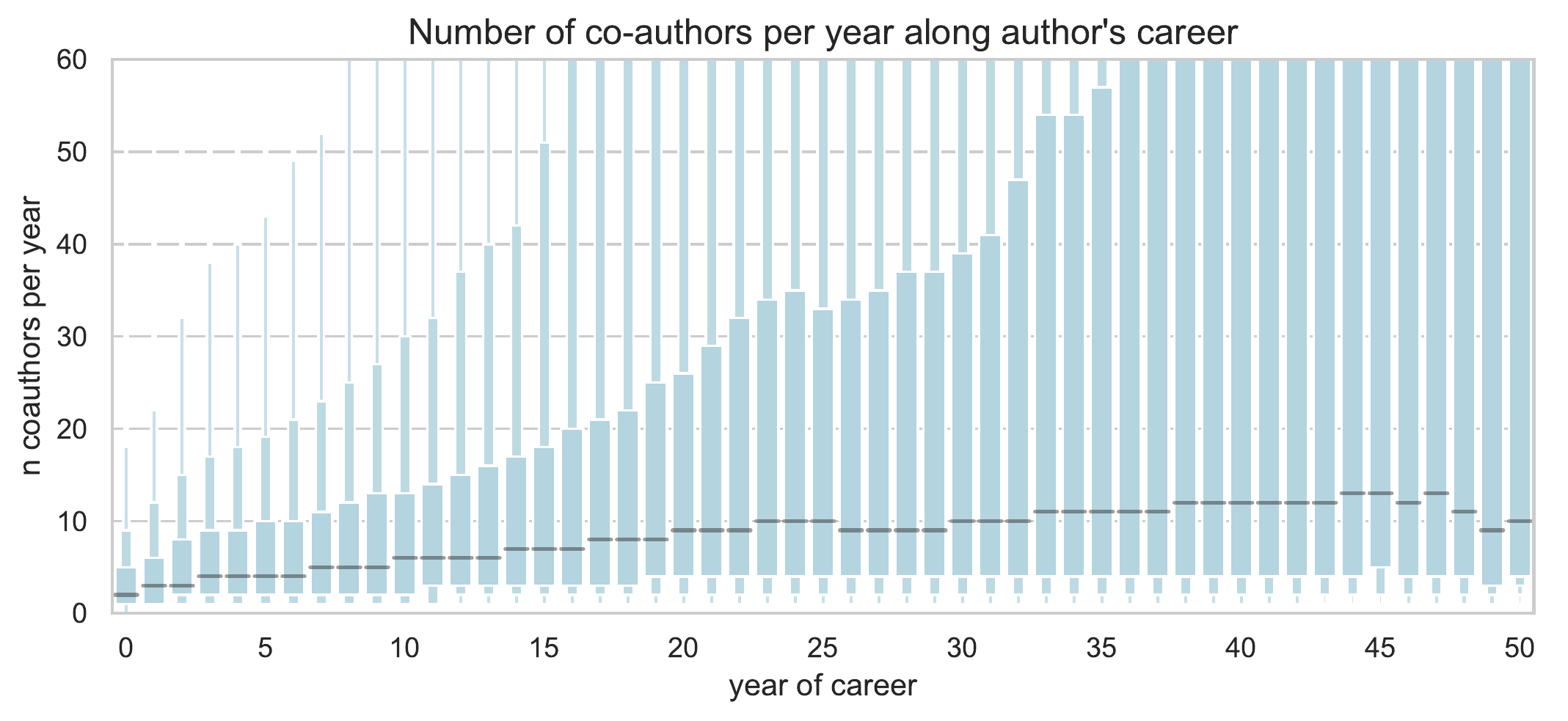}
	\begin{subfigure}[t]{.49\textwidth}
		\includegraphics[width=\textwidth]{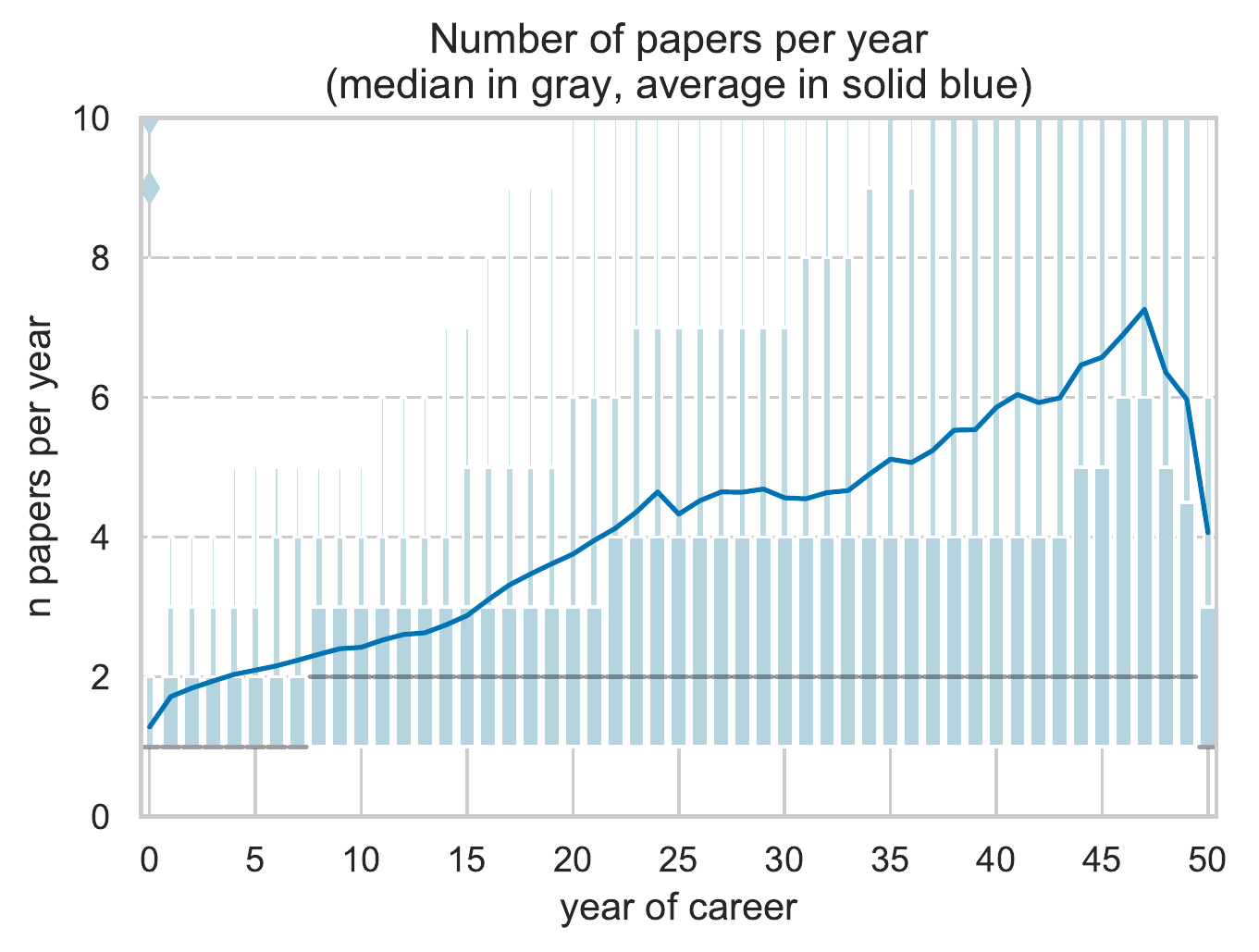}
	\end{subfigure}
	\begin{subfigure}[t]{.49\textwidth}
		\includegraphics[width=\textwidth]{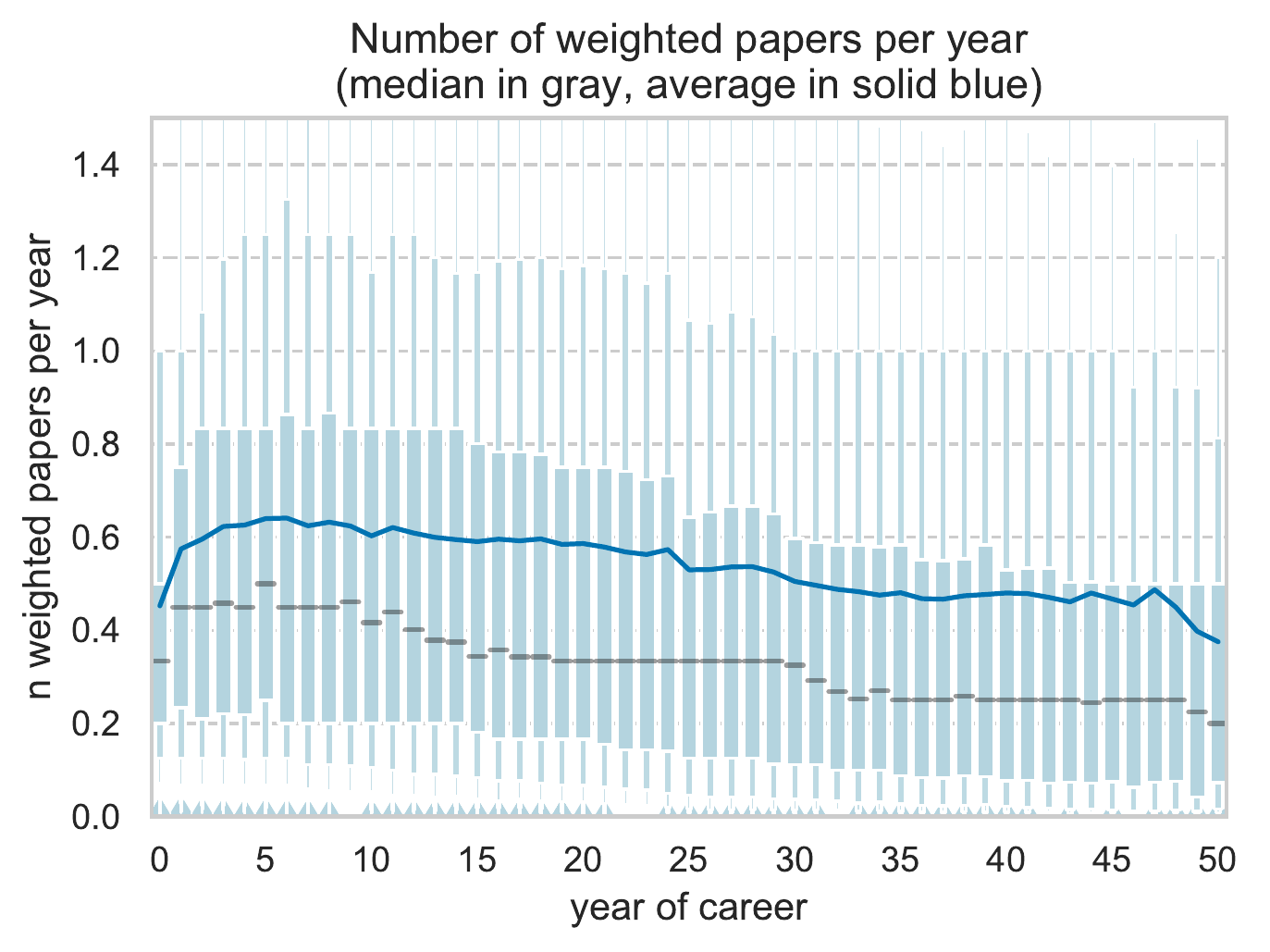}
	\end{subfigure}
	\caption{Trends of individuals' career. Boxplots (a.k.a. letter-value plots, see \cite{hofmann2017value}) where the median is the centerline and the quartiles Q2 and Q3 are the rectangles below and above it, respectively. Each successive level outward contains 50\% of the remaining data. The average is denoted by a solid blue line. A weighted paper is counted as $\frac{1}{\text{n authors}}$.}
	\label{fig:1}
\end{figure}

% \begin{figure}
%     \centering
%     \includegraphics[width=.7\textwidth]{img/3_boxplots_citations_per_5years_along_career.pdf}
%     \caption{Trends of citations to individuals}
%     \label{fig:2}
% \end{figure}

The analysis at an aggregate level seems also to confirm that inefficiencies may be caused by congestion and overproduction of papers.
The aggregate production of papers has significantly increased, together with the sheer number of authors participating to the profession (Figure \ref{fig:3} \begin{rev}in \ref{app:data}\end{rev}). 
However, \begin{rev}Figure \ref{fig:4} shows that\end{rev} the fraction of papers without citations has increased even though the median number of citations received by papers has remained stable over time and the number of papers that got cited -- at least once -- has increased at the same rate as the number of citations produced overall. 
% On the one hand, the aggregate production of papers has significantly increased, together with the sheer number of authors participating to the profession (Figure \ref{fig:3}). 
% On the other hand, although the median number of citations received by papers has remained stable over time and the number of papers that got cited -- at least once -- has increased at the same rate as the number of citations produced overall, the fraction of papers without citations has increased (Figure \ref{fig:4}). 
% Our model offers a possible explanation for these types of inefficiencies (Section \ref{subsec:publishing}).
% Our model offers a possible explanation for these types of inefficiencies: the application in Section \ref{subsec:publishing} already mentioned above not only accommodates the presence of individual-level researchers, but also a society that bears the costs of overproduction of papers and increased size of teams.
Our model offers a possible explanation for how these types of aggregate-level inefficiencies may arise as a consequence of individual-level behavior and non-internalized externalities (Section \ref{subsec:publishing}).

% While the publish-or-perish pressure may explain overproduction of papers and the price paid in terms of their quality, it may also have an unexpected consequence for what concerns the increase in collaborations among scientists.
% Our model provides an insight on the benefits that come from increasing connectedness (Subsection \ref{section:ex}). In the data, it can be seen that the largest connected component of the coauthorship network has increased steadily over time together with the reduction of the average shortest path needed to connect any two authors. This signals a community that not only includes an ever increasing number of participants, but also shortens distances and creates bridges among different sub-communities and sub-fields of physics (Table \ref{tab:stats_coauthorship_network} and Figure \ref{fig:5}).
% The benefits of increased interconnections and of cross-fertilization of fields is also well-documented in the literature \citep{lariviere2015long,pluchino2019exploring,okamura2019interdisciplinarity}.

% \begin{figure}[ht]
%     \centering
%     \includegraphics[width=.55\textwidth]{img/4_n_papers_and_authors_over_time.pdf}
%     \caption{}
%     \label{fig:3}
% \end{figure}

\begin{figure}[ht]
    \centering
     \includegraphics[width=.75\textwidth]{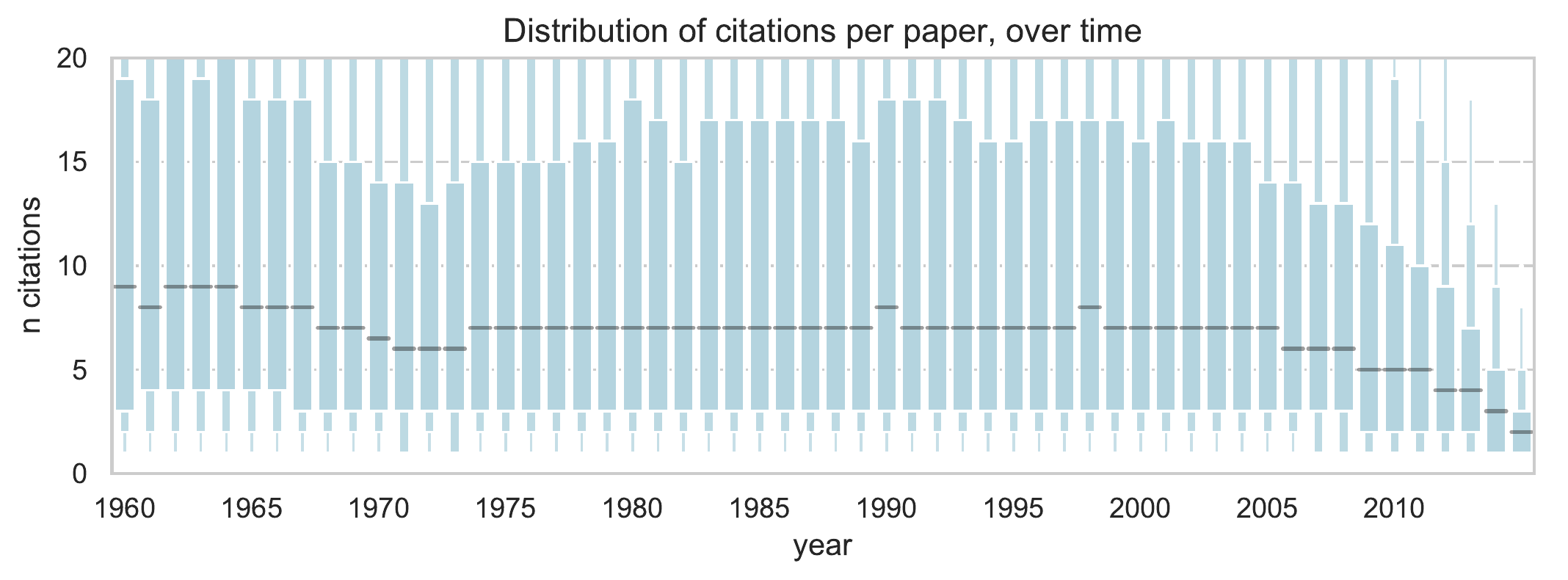}
	\begin{subfigure}[t]{.49\textwidth}
		\includegraphics[width=\textwidth]{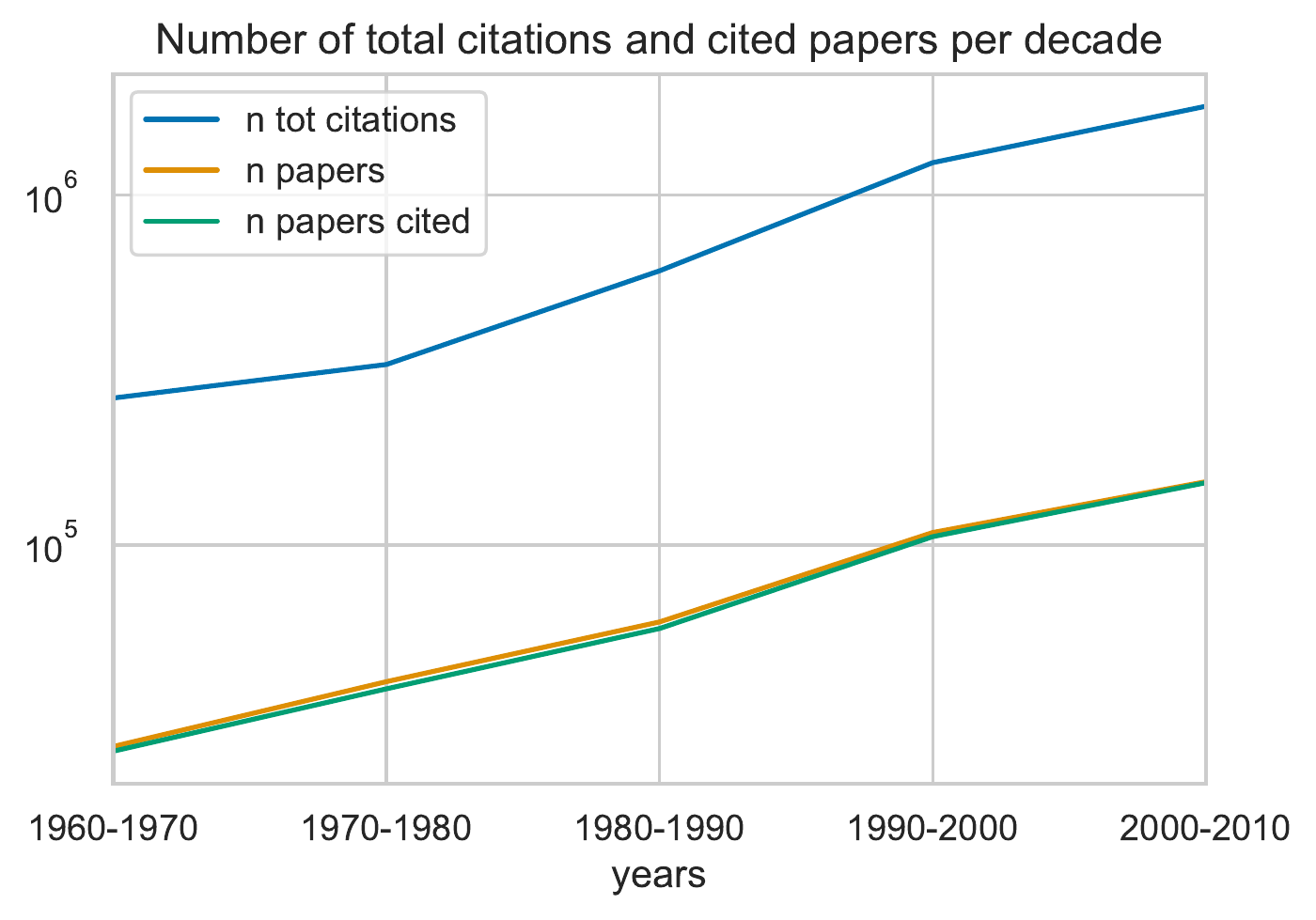}
	\end{subfigure}
	\begin{subfigure}[t]{.49\textwidth}
		\includegraphics[width=\textwidth]{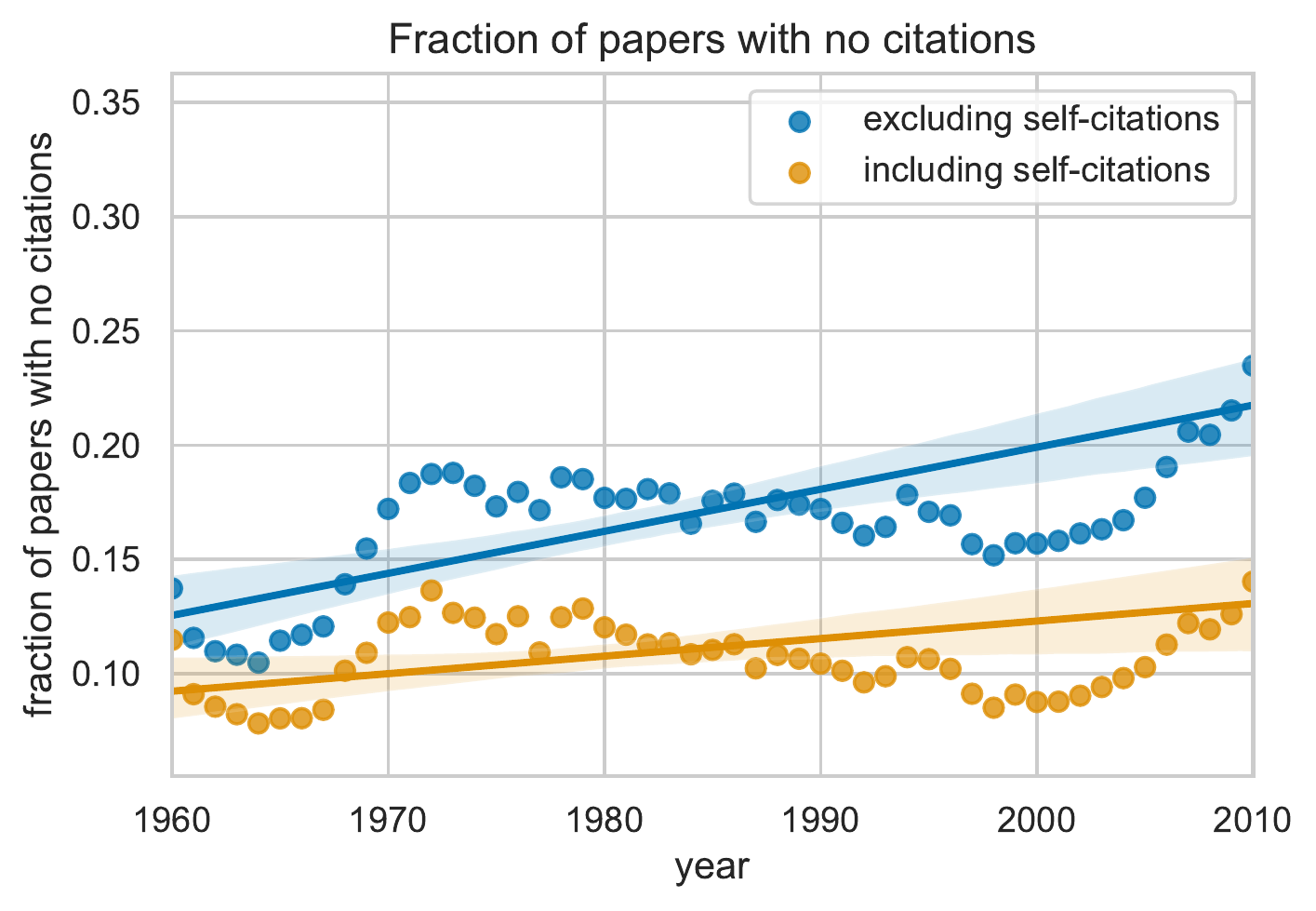}
	\end{subfigure}
	\caption{Trends of citations received by papers over time. The publication year of the cited paper is on the horizontal axis.}
	\label{fig:4}
\end{figure}

% \begin{table}[ht]
% 	\centering
% 	\begin{tabular}{c|r|r|c|c|c}
% 		decade	&	n nodes	&	n edges	&	avg. shortest path length	& diameter	&	density	\\
% 		\hline
% 		1960s	&	14,366	&	45,312	&	7.49	&	23	&	0.000439	\\
% 		1970s	&	27,371	&	128,302	&	6.72	&	20	&	0.000343	\\
% 		1980s	&	49,036	&	448,287	&	5.73	&	18	&	0.000373	\\
% 		1990s	&	94,351	&	1,818,684	&	4.95	&	16	&	0.000409	\\
% 		2000s	&	135,135	&	5,655,490	&	4.38	&	17	&	0.000619	\\
% 	\end{tabular}
% 	\caption{Descriptive statistics of the largest connected component in each decade.}
% 	\label{tab:stats_coauthorship_network}
% \end{table}

% \begin{figure}[ht]
%     \centering
% 	\begin{subfigure}[t]{.49\textwidth}
% 		\includegraphics[width=\textwidth]{img/7b_coauthorship_nodes_edges_over_time.pdf}
% 	\end{subfigure}
% 	\begin{subfigure}[t]{.49\textwidth}
% 		\includegraphics[width=\textwidth]{img/7a_coauthorship_share_in_largest_component.pdf}
% 	\end{subfigure}
% 	\caption{}
% 	\label{fig:5}
% \end{figure}

\end{new_am}

\clearpage

\section{Model}\label{section:model}

\subsection{The team formation model}\label{subsec:game}

We take into consideration a finite set $N$ of $n$ agents. Each agent $i$ has an endowment $w_i \in \mathbb{N}_+$ of a time resource. We denote by $\vec{w} \in \mathbb{N}_+^n$ the vector of endowments of all agents.\footnote{We denote by $\mathbb{N}$ the set of non-negative integers, and by $\mathbb{N}_+$ the set of positive integers.} A \emph{team} is a vector $\vec{t} \in \mathbb{N}^n$, $\vec{t} \leq \vec{w}$, with $t_i$ indicating the amount of time employed by agent $i$ in a joint task. We denote by $T$ the set of teams. 

Let $A$ be a finite set of activities (or tasks). A project $p = (a,\vec{t})$ is an activity $a \in A$ carried out by a team $\vec{t} \in T$. We use set $P \subseteq A \times T$ to collect all projects $p = (a,\vec{t})$ such that team $\vec{t}$ is able to accomplish activity $a$. We can think of $P$ as representing the \emph{technology}, since it indicates, for every possible task, which combinations of inputs allow the task to be completed.\footnote{%
\label{note:Roth}
A complementary interpretation of the technology $P$ is based on the agents' preferences.
From this point of view, $P$ allows only for those teams in which no member would rather stay alone than participate in the project.
In the literature on matching this condition is called \emph{individual rationality} and it is also used in decentralized matching models \citep{roth1990random}. We observe, however, that interpreting $P$ as individual rationality asks for a different model when combining mistakes with exit costs (see \ref{app:robust}): in such a case, a project that is formed by mistake persists over time due to the costs for exiting, even if some agent would prefer to stay alone. 
Finally, we point out that, when allowed, both interpretations for the technology can co-exist: a project is technologically feasible if such a team is able to perform the activity and, at the same time, every agent is willing to do so.} 
It will simplify the following exposition to introduce, with a slight abuse of notation, the auxiliary function $n(p)=n(a,\vec{t}) \equiv \{ i \in N: t_i > 0 \}$, which gives us the set of agents that put some positive amount of time (possibly different among agents) into project $p$.
Another notation we will use is $h(p)=h(a,\vec{t}) \equiv \sum_{i=1}^N t_i $, which indicates the total amount of time (e.g., \emph{hours}) employed on aggregate by the agents in project $p$.
%{\bf [shall we keep $n(\cdot)$ and $i(\cdot)$?]}

In the following discussion we will often use \emph{teams} and \emph{projects} as synonyms, but some clarification is necessary.
A project $p=(a,\vec{t})$ characterizes an activity $a$ performed by a team $\vec{t}$, where $\vec{t}$ specifies not only the members of the team (who are in the set $n(p)$) but also how much time each of them devotes to the project.
\begin{new_am}A collection of projects, i.e., of activities and teams, is called a \emph{state}.\end{new_am}
While every project $p$ can occur only once in a state, because every activity $a$ can be executed only once by the same team, the same team $\vec{t}$ can occur in different projects, if this is allowed by the technology $P$, i.e., if there are at least two projects $(a,\vec{t}),(b,\vec{t}) \in P$, with $a \ne b$.

%A \emph{state} $x$ is defined to be a subset of $P$ where every activity can be performed by one team only, i.e., $p=(a,\vec{t}) \in P$, $p'=(a',\vec{t}') \in P$ and $\vec{t} \neq \vec{t}'$ imply that $a \neq a'$. 
\begin{rev}A \emph{state} $x \subseteq P$ is a collection of projects where every activity is performed by only one team.\end{rev} We use $\vec{e}(x) = \sum_{ (a,\vec{t})  \in x} \vec{t}$ to indicate the vector collecting the overall amount of resources employed in state $x$, agent by agent. We say that $x$ is \emph{feasible} if $\vec{e}(x) \leq \vec{w}$, and we denote by \begin{rev}$X \subseteq \mathcal P(P)$ the collection of subsets of $P$ containing all feasible states.\end{rev} 
We also introduce function $\ell(x)=|x|$ that simply counts the number of projects that are completed in state $x$. 

Finally, we introduce \emph{utilities} that agents earn depending on the state they are in. For every $i \in N$, and for every $x \in X$, we denote by $u_i(x)$ the utility gained by agent $i$ in state $x$.
%Our utility functions are the analogous of \emph{allocation rules} in the network formation literature. [\textbf{secondo me le allocation rules sono per transferable utility}]

Given these elements, it is possible to define a \emph{team formation model} with the quintuple $(N,\vec{w},P,\vec{u})$.
The primitives are the set $N$ of agents involved, their constraints $\vec{w}$, the set $P$ of projects allowed by technology, and agents' utilities $\vec{u}$.
Given $N$, $\vec{w}$ and $P$, it is possible to derive the \begin{rev}set $X$ of all feasible states which is a partially ordered set with respect to set inclusion\end{rev}.\footnote{\begin{rev}The set inclusion is a partial order over the powerset $\mathcal P(P)$ and, hence, on any subset of $\mathcal P(P)$ such as $X$. We recall that a (weak) partial order over a set is a binary relation among the elements of the set that is reflexive, antisymmetric and transitive.\end{rev}}

% Given these elements, it is possible to define a \emph{team formation model} with the quintuple $(N,\vec{w},P,X,\vec{u})$.
% The primitives are the set $N$ of agents involved, their constraints $\vec{w}$, the set $P$ of projects allowed by technology, and agents' utilities $\vec{u}$.
% Given $N$, $\vec{w}$ and $P$, it is possible to derive the partially ordered set $X$ of all feasible states.
%Finally, we define, for every $i \in N$, the value function $u_i:X \rightarrow \mathbb{R}$.

\subsection{Assumptions} \label{subsec:assumptions}

In deriving our results, we employ the following restrictions on the possible structure of teams (first three) and on utilities (second group of three). We explicitly refer to each of these assumption whenever used. We note that some of them are a refinement of one another, while others are incompatible. 

\renewcommand{\theassumption} {t1}
\begin{assumption}
In every $(a,\vec{t}) \in P$, we have for every $i \in N$ that either $t_i= 0$ or $t_i=1$.
\label{ass:t1}
\end{assumption}

Assumption \ref{ass:t1} states that the time allocated to each feasible project by every agent is always $0$ or $1$, or simply (up to a normalization of time) that the time allocated to each feasible project by its participants is a constant of the model which is homogeneous across projects for every agent. 

In contrast, the next two are assumptions that exogenously fix the number of members in each team. We will discuss them in more detail in Section \ref{subsec:generalization} where we will see how our model is a generalization of other common theoretical setups.

\renewcommand{\theassumption} {s1}
\begin{assumption}
There is a $k \in \mathbb{N}_+$, such that for every $p \in P$, we have that $|n(p)|=k$.
\label{ass:s1}
\end{assumption}
The following Assumption \ref{ass:s2} is a refinement of Assumption \ref{ass:s1}, where $k$ is fixed to be equal to $2$.
\renewcommand{\theassumption} {s2}
\begin{assumption}
For every $p \in P$, we have that $|n(p)|=2$.
\label{ass:s2}
\end{assumption}
We now present some assumptions that specify how agents gain utilities by performing activities in teams.
\renewcommand{\theassumption} {v0}
\begin{assumption}
For every $x, x' \in X$, with $x'\neq x$ and $x' = x \cup \{p\}$, and for every $i \in N$ such that $i \in n(p)$, we have that $u_i (x') > u_i (x)$.
\label{ass:v0}
\end{assumption}
Assumption \ref{ass:v0} is the only one that is needed for our main result. It states that the marginal utility in forming a team, for each of its members, is always positive, independently of all other teams in place. 
We note that this assumption allows for a large variety of externalities that a project may have on the utility of non-members of that team, or on the fact that the same team could bring different marginal effects to its members, depending on the state.
\begin{new_am}
In particular, \begin{rev} this is in line with the assumptions of the model in \cite{baumann2021model}, where the benefit of an agent from participating to a project always increases if she gets involved in it.
Moreover, this assumption is consistent with the behavior of researchers that we observe in the APS dataset discussed in Section \ref{subsec:empirical_motivation}.\end{rev}
\end{new_am}

An additional possible restriction is to impose that the aggregate utility of each project is constant across projects (normalized to $1$).

\renewcommand{\theassumption} {v1}
\begin{assumption}
For each $x \in X$, $\sum_{i \in N} u_i (x) = |x|$.
\label{ass:v1}
\end{assumption}

Finally, we will consider also a more restrictive assumption that asks for linearity in teams membership, so making the marginal value of each team, for each of its members, independent on states.

\renewcommand{\theassumption} {v2}
\begin{assumption}
For each state $x \in X$, and any agent $i \in N$, we have that $u_i (x) = v \cdot |\{p \in x: i \in n(p)\}|$, with $v \in \mathbb{R}^+$.
\label{ass:v2}
\end{assumption}

The last two assumptions convey different ideas on the assignment of utilities: while Assumption \ref{ass:v1} imposes that the aggregate marginal value of each team is $1$, Assumption \ref{ass:v2} says that the payoff earned by each agent $i$ is merely given by the number of projects in which $i$ participates.
We note that the two assumptions are compatible only if Assumption \ref{ass:s1} holds as well, in which case we have $v=\frac{1}{k}$.

\subsection{Maximal states} \label{subsec:maxstates}

We observe that $X$ is a partially ordered set under inclusion.
This is because, for any two states $x$ and $x'$ belonging to $X$, we can have that either $x$ is included in $x'$, or $x'$ is included in $x$, or no set inclusion relationship can be established between them.
However, as the empty state $x_0$ is included in any other state, it is the only \emph{minimal} state (or the \emph{least} state) and, given two  states $x$ and $x'$, the set of those states that are included in both is always non-empty.
On the other hand, as there is a threshold $\vec{w}$ on the overall available resources, there may not always be a common superset for any two states.
In general there will be many \emph{maximal} states, i.e., states above which it is not possible to include other teams, because otherwise the threshold would be exceeded.

We denote by $\mathcal{M}$ the set of maximal states, $\mathcal{M} = \{x \in X : x \subseteq x' \text{ and } x \neq x' \Rightarrow x' \notin X \}$. We denote by $\mathcal{L}$ the set of states with maximum number of completed projects, $\mathcal{L} = \{x \in X : |x| \geq |x'|, \text{ for all } x' \in X\}$. We observe that $\mathcal{L} \subseteq \mathcal{M}$. In fact, if $x \in X$ and $x \notin \mathcal{M}$, then there exists a feasible state that can be obtained from $x$ by adding some project, and $x$ cannot maximize the number of projects. In contrast, there exist in general maximal states that do not maximize the number of projects, as the following example shows.

\begin{example}[Maximal states and maximum number of projects] \label{ex:POset}
Consider the case in which $N=\{i,j,k,m\}$, $\vec{w}=(2,2,2,2)$, $A = \{a,b\}$, and $P = \{ (a,(1,1,0,0)), (b,(1,1,0,0)), (a,(0,1,1,0)),$\\ $(b,(0,1,1,0)),$ $(a,(0,0,1,1)), (b,(0,0,1,1)) \}$.
This is a situation in which there are four agents with two units of time each, there are two activities to be performed, and each activity requires that either $\{i,j\}$, or $\{j,k\}$, or $\{k,m\}$ must be involved, with one unit of time each. We note that Assumptions \ref{ass:t1} and \ref{ass:s2} hold.
Figure \ref{fig:POset} illustrates the partial order on set $X$ resulting from the above assumptions: an arrow from a state $x$ to another state $y$ indicates that we can pass from $x$ to $y$ by adding a single project.\footnote{In order to provide a simplified figure, we have summarized in a single node the states that are the same in any respect apart from the labels of activities.\label{foot1}}
There are three maximal states, but only one of them maximizes the number of projects.
\eproof
\end{example}

\begin{figure}
\begin{center}
\includegraphics[width=1\textwidth]{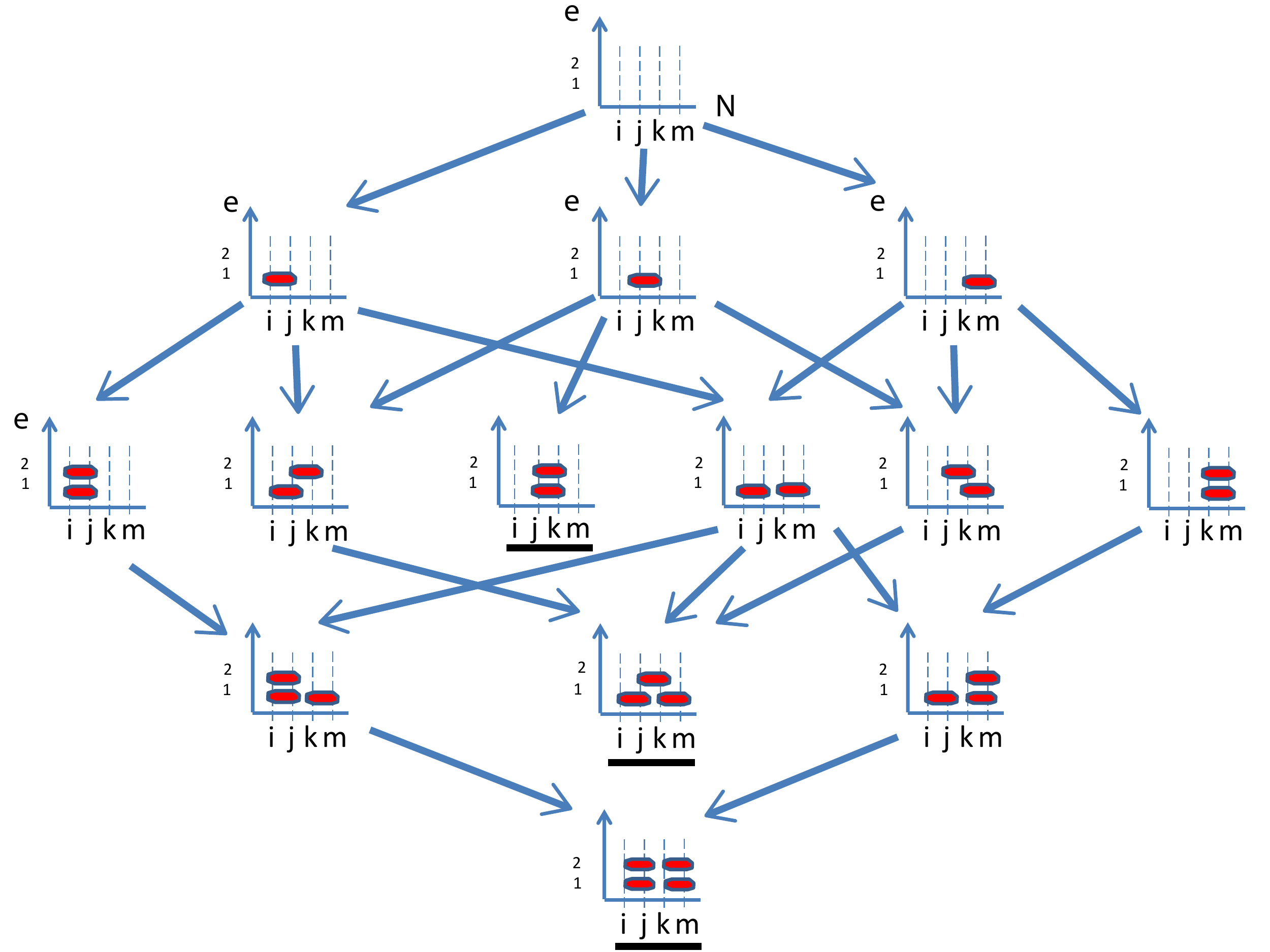}
\caption{%
The partially ordered set $X$ of all feasible states for the team formation model of Example \ref{ex:POset}.
In this graphical representation -- where we have not considered differences in activities, focussing only on their number -- there are three maximal states, which are those underlined in black.}
\label{fig:POset}
\end{center}
\end{figure}

\subsection{Why this generalization?} \label{subsec:generalization}

The theoretical setup we have introduced encompasses different matching models with non-transferable utility that have been developed in the literature. We acknowledge that some of the existing models might be stretched to deal with most of the cases analyzed within our setup. However, we claim that our model is a natural and simple container for all these models, and we find a value in its capability to adapt easily so as to take into consideration specific cases. In the following we will illustrate this capability.

% The knapsack problem (introduced in Section \ref{subsec:literature}) is obtained if we have only one agent who can carry an overall weight of $w$ (the resource endowment) in her knapsack, and she can choose objects (or activities) in set $A$, where every object $a \in A$ has value $v_a$ and weight $w_a$. A project is here an object with its weight, i.e.,  $(a,w_a)$. A state is a set of objects within the knapsack, such that the sum of weights does not exceed $w$. The agent's utility in state $x$ is the sum of values of the objects in the knapsack.    

Cooperative games with non-transferable utility are obtained in our setup if we specify that each agent can belong to one coalition only, and that no externalities are allowed. In order to deal with matching, as done in \cite{BonPin18}, we simply need to add Assumption \ref{ass:s2}, so that only teams of size two are allowed to be formed. Marriage -- that is bipartite matching -- can be obtained through adequate constraints on the technology; after dividing the set of agents between males and females, only heterosexual pairs are allowed in $P$, and additional constraints can also be considered. 
%The formalism that we have introduced generalizes other well known \emph{matching} environments.
%To begin with, if we set $w_i=1$ for each agent $i \in N$, and we impose Assumption \ref{ass:s2}, we have the setup of standard matching models, with the constraints eventually imposed by $T$.
%As an example, the set of agents could be bipartite between \emph{males} and \emph{females}, and only heterosexual links would be allowed by the \emph{technology}, or maybe only a subset of them.
Figure \ref{fig:marriage_problem} provides an example.
%We will discuss in detail this example and a general application of our setup to marriage in Section \ref{section:am}.

\bigskip

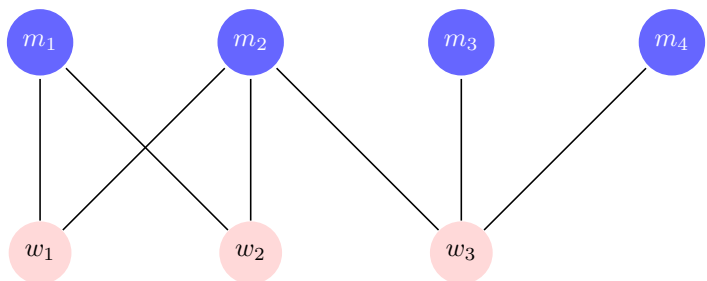
\begin{figure}[h]
\begin{center}
\begin{tikzpicture}[-,>=stealth',shorten >=1pt,auto,node distance=2.8cm,semithick]
  %\tikzstyle{state}=[circle,fill=blue!30!yellow,draw=none,text=-blue!30!yellow, minimum size=12mm]
  \tikzstyle{F}=[circle,fill=pink!60!white,draw=none,text=black, minimum size=5mm]
  \tikzstyle{M}=[circle,fill=blue!60!white,draw=none,text=white, minimum size=5mm]
  \node[F]            (W1)                      {\footnotesize{$w_1$}};
  \node[F]         (W2)  [right of=W1]          {\footnotesize{$w_2$}};
  \node[F]         (W3)  [right of=W2] 		 {\footnotesize{$w_3$}};
  \node[M]         (M1)  [above of=W1] 		 {\footnotesize{$m_1$}};
  \node[M]		   (M2)  [right of=M1]		 {\footnotesize{$m_2$}};
  \node[M]		   (M3)  [right of=M2]		 {\footnotesize{$m_3$}};
  \node[M]		   (M4)  [right of=M3]		 {\footnotesize{$m_4$}};
  \path (W1)  edge 			   (M1)
		(W1)  edge 			   (M2)
		(W2)  edge 			   (M1)
		(W2)  edge 			   (M2)
		(W3)  edge 			   (M2)
		(W3)  edge 			   (M3)
		(W3)  edge 			   (M4);
\end{tikzpicture}
\end{center}
\caption{A marriage problem. In this example there are three women ($w1$, $w2$ and $w3$) and four men (from $m1$ to $m4$). Feasible pairs are joined by a link. 
\label{fig:marriage_problem}}
\end{figure}

If, instead, we relax the upper bound on $\vec{w}$, still following Assumption \ref{ass:s2}, then any state can be considered as a network that satisfies the constraints imposed by $\vec{w}$ (concerning the maximum degree of nodes), and the connections made available by technology $P$ (representing the exogenous network of opportunities).
This is illustrated in the following example.

\bigskip

\begin{example}[Co-authorships] \label{ex:coauthors}
A popular and seminal model in the economic literature on networks is the ``Co-author model'' of \citet{JacksonWolinsky96} \begin{rev}and later extended by \cite{rego2019co}\end{rev}:\footnote{\begin{rev}With the notation used here, the model by \cite{JacksonWolinsky96} and by \cite{rego2019co} consider activities that are performed only by pairs of agents, as stated in Assumption \ref{ass:s2}.\end{rev}} agents are researchers and links are pair-wise collaborations on scientific projects, which are costly but provide payoffs that depend endogenously on the negative externalities given by each collaboration to the other co-authors of an author.

We can include in our setup a payoff function with costs and negative externalities of a project $p=(a,\vec{t})$ on the members of other teams that are formed by the members of $\vec{t}$, as in the original model\begin{rev}, and with Assumption \ref{ass:s2}\end{rev}.
However, our model allows for more generality and also for more realistic time constraints\footnote{%
It can be reasonable to assume (at least in certain contexts) that everyone likes to be in as many projects as possible, and hence that Assumption \ref{ass:v0} is satisfied. With this interpretation, costs arise indirectly as opportunity costs due to the constraint given by available time.}
that can be imposed on the available (multi-)matchings.
First of all, (i) some agents may work alone, but even three or more agents can set up a team together and produce a paper, as happens in the profession.
Then, (ii) with regard to constraints, there could be an exogenous network $G$ of acquaintances, so that a group of co-authors is possible only if they are mutually connected in $G$.
Or, (iii) the researchers could have exogenous complementary skills, and only projects involving agents with enough diversity could be successful.
Aspects such as the three listed above, and even others, could all be modeled by some technology $P$.
%Finally, the fact that in our model the same team can be replicated, this could be interpreted as the (discrete) amount of effort that members put into each project, and this effort can be heterogeneous inside a team because successful teams allow for agents contributing heterogeneously. 
\eproof
\end{example}

So, what is the added value of our setup with respect to existing ones, in terms of representation of real-world phenomena?
To provide an answer through an example, let us stick to the co-authorship model of Example \ref{ex:coauthors}.
Imagine that agents $i$, $j$ and $k$ set up a project $p $ together, so that $|n( p )|=3$. This could be represented in the original co-authorship model of \citet{JacksonWolinsky96}, that allows only for couples, by saying that $i$ is linked to $j$ and $k$, and $j$ and $k$ are also linked together.
We observe that the link between $i$ and $j$ would have a negative externality on each neighbor of these agents, including $k$.
However, in the general setup and in reality, the fact that $i$ and $j$ are in a three-agent collaboration has a positive externality on the third agent $k$, and a negative externality on the others.
Formally, this could be done in the original network formation model by specifying, for each link between two agents, and for any other agent, the sign of the externality of that link on the third agent.
It is clear that this would seriously complicate the notation, and that only a more general framework such as the one we use can overcome such difficulties.

\bigskip

We conclude this section with a stylized but fairly general applications that show how the possibilities and the competing incentives of some environments cannot be dealt with using the standard models of matching and network formation.
\begin{rev}
This model provides a possible mechanism for explaining some of the inefficiencies that we observe in the APS dataset discussed in Section \ref{subsec:empirical_motivation}, where researchers seem to congestion the overall activity because they do not  internalize the negative externalities that they have on each other.\end{rev}

\subsection{The publishing model}
\label{subsec:publishing}

We present here an extended example providing the general idea of the model.
We continue to adopt an intuition related to \begin{new_am}the insights presented in Section \ref{subsec:empirical_motivation} and to\end{new_am} the daily experience of everyone in the academic profession, but it is clear that it can easily be extended to R\&D between firms that are competing in a market, as in the model of \cite{GJ03}.
Consider a world where there are $n$ homogeneous scientific authors, each trying to form teams of collaborators and each with a common time constraint $w$.
They all have two goals: a good output in terms of publications (on which they compete with colleagues), but also the objective of doing good research that can provide advancements in the field. 
Each author maximizes in each project both the probability of being published and the probability of authoring a good idea.
We assume no constraint on the multi-matching technology $P$, except for the fact that agents cannot work alone: $p \in P$ if and only if $|n(p)| \geq 2$. 

%We write $\vec{t} \bowtie \vec{\tau}$ to say that the two teams share at least one member together.
Here we assume that a project $p$ has a strictly positive \emph{divulgative fitness} (i.e., an expected popularity) that we call $\phi(p)$.
The divulgative fitness of a paper may depend on the amount of work that is put into the paper by its members.
%, but also (positively) on how many other papers they are working on (because of popularity reasons -- it can be a reduced form of reputation in our non-dynamic setting).
This fitness can clearly also be related to heterogeneous exogenous factors.
%This fitness measure capture the fact that work devoted to it is important, but in terms of making it known there are also positive externality if its coauthors also have put a lot of work in other projects.

Accordingly, a paper's \emph{probability of being published} has the multinomial form:\footnote{We set up the model as if only one paper is \emph{published}, but this can be easily relaxed as long as the number of published papers is fixed and independent of the aggregate fitness.}
\[
P_{pub} (p) = \frac{ \phi (p)}{ \sum_{q \in x} \phi (q) }.
\]
When a new team is formed there are clear negative externalities (increasing in the divulgative fitness of the new project) for all the agents that are not members of the new team, because their probabilities of being published decrease. 

In a related but not necessarily collinear way, we assume that each project has a strictly positive probability of providing a good idea which is $P_{good} (p) $.
This probability is reasonably increasing in effort, but there may be communication and coordination costs which make it decrease in the size, in terms of members, of the team. Or, there could be positive externalities from the aggregate quality of all the scientific production as a whole. 
In general, the whole environment of $x$ can provide both positive and negative externalities, with network effects like those described in the \emph{connection model} and in the \emph{co-authorship model} of \cite{JacksonWolinsky96}.

To provide a simplified functional form, which maintains the general idea, we assume that each author $i$ receives a payoff in a generic state $x$ that is:
\[
u_i (x) = \sum_{p: i \in n(p)} 
\left( U P_{pub} (p) + \frac{V}{|n(p)|} \cdot P_{good} (p) \right),
\]
where $U$ and $V$ are positive numbers, homogeneous for all agents,\footnote{%
In a context of industrial organization, with R\&D between firms, $U$ could be the aggregate value of a fixed market, on which firms compete for shares, while $V$ could be the expected value of further markets that new products could open.}
and $P_{good} (p)$ does not depend on other existing projects in $x$. We observe that, while the utility $U$ coming from a publication is not affected by the number of authors (what matters is to have a publication in the curriculum vitae), the benefits $V$ deriving from a good idea must be shared among the participants (consider, for instance, the earnings that come from a patented idea).

\begin{new_am}
This utility is in line with what observed in Subsection \ref{subsec:empirical_motivation}. On the one hand authors benefit from taking part in many projects and from having multiple collaborators to accommodate and meet the time constraint while, on the other hand, they do not take into account the externalities which may cause a reduction in effort spent and on the quality of the project.
\end{new_am}

\bigskip

For this simplified model it is not difficult to prove that it satisfies Assumption \ref{ass:v0} \begin{rev}(while the other assumptions are not)\end{rev}.\footnote{%
It is possible to provide more complicated payoff functions, which are non-linear or for which $P_{good} (p)$ is state dependent, and which also satisfy Assumption \ref{ass:v0}.
A simple but reasonable first step of generalization is by assuming the existence of non-negative net externalities on $P_{good} (p)$ that come from a member of $n(p)$ that participates in other projects. This would bring our model closer to the \emph{connection model} than to the \emph{co-authorship model}. Indeed, the negative externalities of the \emph{co-authorship model} are an indirect way to take into account the scarcity of time that researchers face in their activity; the reason for such externalities is removed in our setting, where agents are explicitly given time endowments.}
That is because, for an agent $i$, if we call $\Phi \equiv  \sum_{q \in x: i \in n(q)}  \phi (q)$, the marginal utility for being member of a new  project $p'$ is:
\[
u_i (x \cup p') - u_i (x)
 = U \left( \phi(p') \frac{ \sum_{q \in x: i \not \in n(q)}  \phi (q) }{\left( \Phi + \phi(p')\right) \Phi }\right)+ V \cdot P_{good} (p').
\]
The first term is non-negative, and it is null only if that agent was already a member of each existing team.
The second term is strictly positive by definition.

\section{Myopic team-wise stability}\label{section:mts}

The first equilibrium notion that we provide -- called \emph{myopic team-wise stability} -- is a direct generalization of the concept of \emph{pair-wise stability}, from \citet{JacksonWolinsky96}, which is used in network formation games\begin{rev}, such as \citet{JacksonWatts02}. The original notion only considers activities performed in pairs while the present extension allows for activities performed by groups of different sizes\end{rev}.\footnote{\begin{rev}More precisely, in the model by \cite{JacksonWatts02} Assumption \ref{ass:s2} always holds and, in the language of \cite{rego2019co}, agents' capacity of forming links and performing tasks is assumed unbounded.\end{rev}}

\subsection{Myopically team-wise stable states}\label{section:def_mts}

The following definition formalizes the notion:
\begin{definition}
A state $x$ is \emph{myopically team-wise stable} {\bf [MTS]} if 
\begin{itemize}
\item[(i)] for any project $p \in x$, and for any agent $i \in n(p)$, we have that $u_i (x) \geq u_i (x \backslash \{  p  \}) $;
\item[(ii)] %for any project $ p$, such that $x \cup \{ p \} \in X$, if there exists an agent $i \in N$ such that $i \in n(p)$, and $u_i (x \cup \{ p \}) > u_i (x) $, then there exists another agent $j \in N$ such that  $j \in n(p)$, and $u_j (x \cup \{ p \}) < u_j (x) $.
there exists no project $p \in P$ such that $x \cup \{ p \} \in X$ and, for any agent $i \in n(p)$, $u_i (x) \cup \{ p \} > u_i (x)$.
\eproof
\end{itemize}
\end{definition}

In words, a state $x$ is \emph{myopically team-wise stable}, if $(i)$ there is no agent that would be better off by deleting a project she belongs to; and $(ii)$ there is no project that could be added to state $x$, without hitting the constraints of its members, and which would make them all strictly better off.\footnote{We require a strict Pareto improvement for the members of a project that is going to be formed in order to conclude that a state is not myopically team-wise stable, while a weak Pareto improvement is usually considered sufficient. We remark that our choice -- which in principle yields a weaker equilibrium concept -- makes no actual difference if Assumption \ref{ass:v0} is satisfied.\label{foot2}} 
With some abuse of notation, we also denote by $MTS$ the set of states that are myopically team-wise stable.

The reason why we call it myopic is that it considers only deviations of one single step in the partially ordered set $X$.\footnote{%
The network concept of pair-wise stability is likewise \emph{myopic}.
On this see \cite{PWK05} and discussion in \cite{kirchsteiger2016limited}.}
Note also that, even if a dynamic is implicit in the definition, this concept of equilibrium is a static one.
% Again, for comparison with and discussion on the network formation case see \citet{JacksonWatts02}.

\bigskip

%Now we introduce the only assumption that we will maintain throughout the whole paper, and see how this will give us a simple characterization of myopically team-wise stable states.
The following Lemma is the first building block of our results.

\begin{lemma}\label{lem:myopic=maximal}
Take a team formation model satisfying Assumption \ref{ass:v0}. We have that $MTS = \mathcal{M}$.
\end{lemma}

\begin{proof}
%The proof relies on the fact that $ u_i ( x \cup \{p \} ) > u_i (x) $ if $i \in n(p)$. \\
If a state is maximal it is not possible to add a project, and any deletion would damage the members of the removed project (by Assumption \ref{ass:v0}). 
Therefore, that state is $MTS$.%, and the set of maximal states $\mathcal{M}$ coincides with the set of myopically team-wise stable $MTS$. \\

Suppose that a state is not maximal, then it would be possible to add a project, which would add a positive marginal amount to the utility of all its members (by Assumption \ref{ass:v0}), and thus that state is not myopically team-wise stable.
\end{proof}

\subsection{Direct and indirect externalities} \label{subsec:inefficiencies}

Our specification allows for externalities between the agents, or for a non-trivial structure of preferences of the agents toward the other team members.
All the inefficiencies arising because positive and negative externalities are not endogenized by the agents would give rise to a comparison between stable and efficient outcomes that would be very similar to the one extensively analyzed in network formation models (see \citealp{Jackson05}).

However, even if utilities from states had a simple structure, e.g., as in the case of the one imposed by Assumption \ref{ass:v2}, numerous indirect effects would arise from the constraints imposed by the technology $P$, and by the vector $\vec{w}$ of endowments, as will be clear from the following examples.

%\bigskip

First of all, consider again the case illustrated in Example \ref{ex:POset} and Figure \ref{fig:POset}.
Because of the constraints imposed by the technology, agents $j$ and $k$ clearly have a negative externality on the other two agents when they form a team together: by forming a team on an activity they reduce the available teams for agents $i$ and $m$.

\begin{example} \label{example:indirect_externalities}
Consider a team formation model with four agents: $i$, $j$, $k$ and $m$, all with an endowment of $3$ units of time.
Agent $i$ can form a team with $j$ only, if one of the two puts in $1$ unit and the other $2$ units of time.
The same holds for the couple formed by $k$ and $m$.
Agents $j$ and $k$, on the other hand, can form a team together by investing only one unit of time each.
As illustrated graphically in Figure \ref{fig:burgerking_s}, this team formation model has three myopically team-wise stable sets: $I$, $II$ and $III$.\footnote{We use here the same simplifying representation that we have employed in Figure \ref{fig:POset} and briefly discussed in footnote \ref{foot1}.}
The payoff of a project for an agent is always $\frac{1}{2}$.
Therefore, this example does not satisfy Assumption \ref{ass:t1}, but satisfies Assumptions \ref{ass:s2} and \ref{ass:v2}. 

In this case, even if all the teams have the same payoff for each of their members, there are indirect preferences for some of the agents.
In particular, if agents $j$ and $k$ could \emph{choose ex ante} their teams they \emph{would rather} form teams together, because this would allow them to form up to $3$ teams, while forming a team with $i$ or $m$ (respectively) would bind them to myopically team-wise stable sets where they can form at most two teams.
We will formalize in Section \ref{section:cs} the idea of `\emph{choosing ex ante}' and `\emph{would rather}'.
\end{example}

\begin{figure}[h]
\begin{center}
\includegraphics[height=10cm]{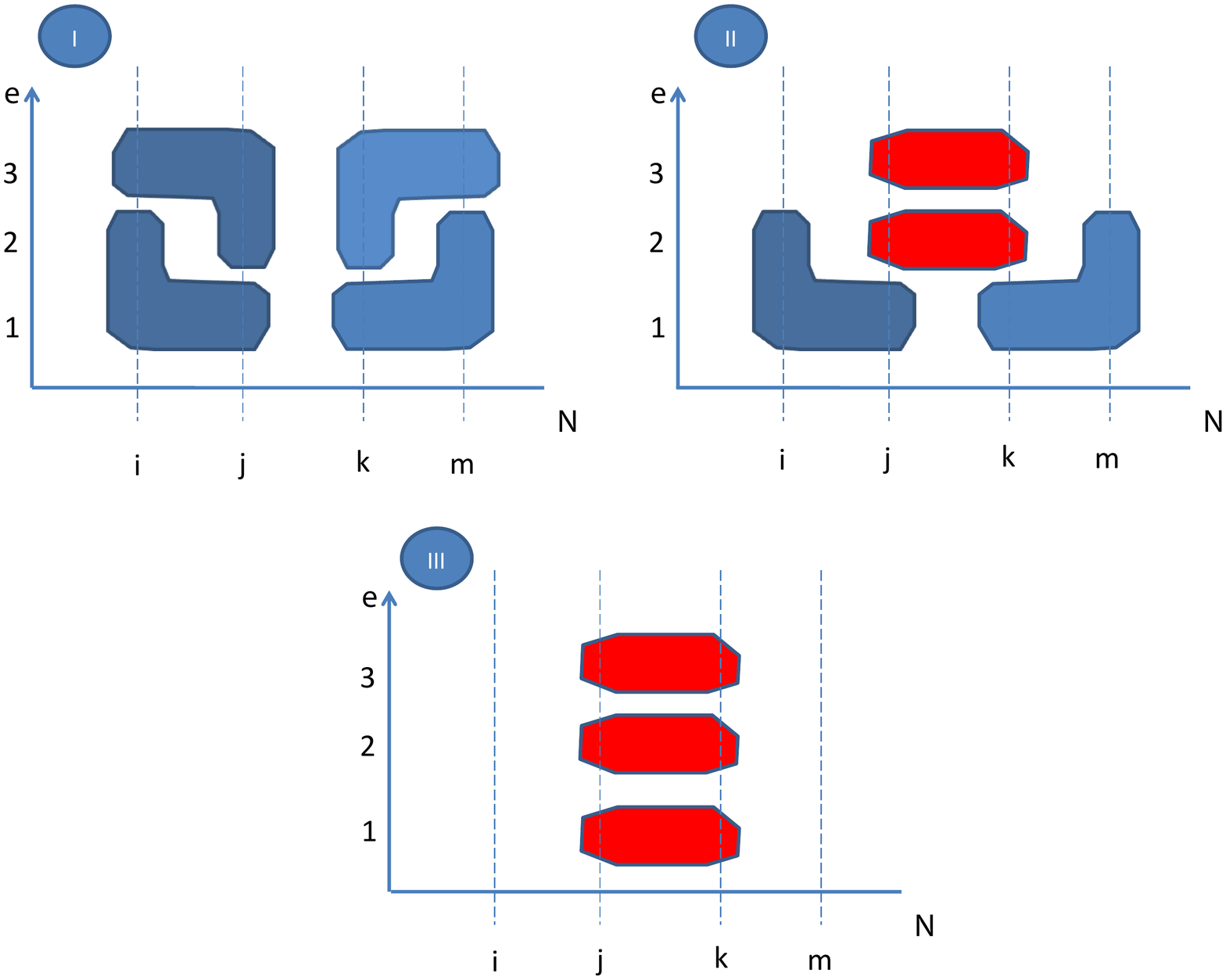}
\caption{$MTS$ states of Example \ref{example:indirect_externalities}.}
\label{fig:burgerking_s}
\end{center}
\end{figure}

\begin{example} \label{example:ij_jk}
As the model is specified, it is not possible to have indirect externalities of the following form: a team is feasible if and only if another team is not, and possibly also the other way round.
As an example consider Case I in Figure \ref{fig:hijkl}, where we may want to express a condition by which the team $(i,j)$ can be active only if the team $(k,m)$ is not active.
Another example could be one in which the same subset of agents cannot simultaneously work on more than one project, even if many are independently feasible.
% (and this is a case we will discuss in Subsection \ref{section:ex}).
It is possible however to model a similar situation including fictitious agents, with a dummy utility function, whereby the agents have limited resources of time and must be included as members in the teams that we want to be mutually exclusive.
We will not develop all the formal definitions of this approach, but we maintain the simple example given in Figure \ref{fig:hijkl}:
as illustrated in Case II it is possible to add a fictitious agent $h$ with $w_{h}=1$, and such that the possible teams are now  $(i,j,h)$ and  $(k,m,h)$.
In this way the original two teams become mutually exclusive. \begin{rev}Adding such fictitious agents does not affect the analysis nor the results and does not require any further assumption.\end{rev}
\end{example}

\begin{figure}[h]
\begin{center}
\includegraphics[height=8cm]{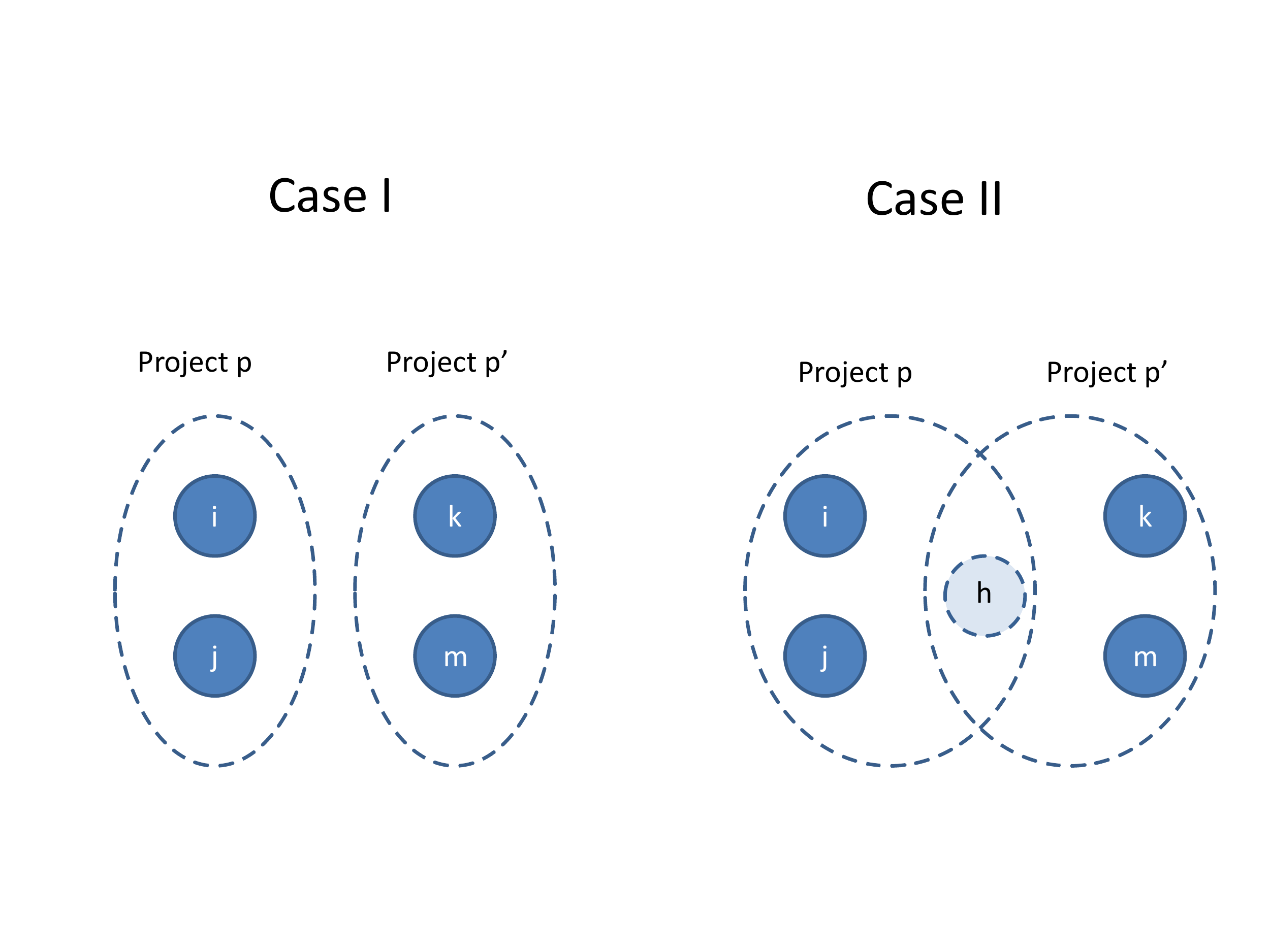}
\caption{The two cases discussed in Example \ref{example:ij_jk}.}
\label{fig:hijkl}
\end{center}
\end{figure}

\section{Stochastic stability}\label{section:ss}

To refine $MTS$ we consider an unperturbed dynamics where absorbing states are myopically team-wise stable states, and we then insert vanishing perturbations with the aim of refining our prediction by means of \emph{stochastic stability}.

\subsection{Unperturbed dynamics and preliminary results}
\label{subsection:ss_unperturbed}
In order to deal with stochastic stability, we need to introduce an underlying dynamics which describes the probabilistic passage from state to state, and then to add perturbations. In particular, we work with discrete time and we indicate it with $s = 0,1,\ldots$. We denote the state of the system at time $s$ with $x^s$. At time $s+1$ a single project $p \in P$ is drawn, with every project in $P$ having positive probability of being drawn.\footnote{See \citet[][Section~7.2]{newton2018evolutionary} for a recent survey of evolutionary models.} 

One remark is worth making at this point.
Extensive heterogeneity is allowed between the probabilities of different projects being drawn; for instance, it might be reasonable to assume that better projects (in some sense) are more likely to be selected and then implemented. As long as every feasible project has a positive, even tiny, probability of being drawn all our results remain valid.

%Then, by limiting the extraction to elements of $P$, we are assuming that teams that are unable to complete the task will not be formed.\footnote{%Alternatively, we may think of them as forming and then dissolving since unproductive. See also Note \ref{note:Roth} for a discussion on what is a \emph{feasible} project.} 
The extracted project $p$ will actually be formed if $x \cup \{p\}$ is feasible, i.e., $\vec{e}(x \cup \{p\}) \leq \vec{w}$. Otherwise, such a team is not formed, since its creation is not possible due to resource constraints. In any case, state $x^{s+1}$ is reached. We refer to this dynamic process as \emph{myopic team-wise dynamics}.

A Markov chain $(X,D)$ turns out to be defined, where $X$ is the state space and $D$ the transition matrix, with $D_{xx'}$ denoting the probability of moving from state $x$ to state $x'$. We recall some concepts and results in Markov chain theory, following \citet{youn1}. Given any two states $x,x' \in X$, state $x'$ is said to be \emph{accessible} from state $x$ if there exists a sequence of states starting from $x$ and reaching $x'$ such that the system can move with positive probability from each state to the next state. A set $\mathcal{E}$ of states is called an \emph{ergodic set} (or \emph{recurrent class}) when each state in $\mathcal{E}$ is accessible from any other state in $\mathcal{E}$, and no state outside of $\mathcal{E}$ is accessible from any state in $\mathcal{E}$. If $\mathcal{E}$ is an ergodic set and $\vec{x} \in \mathcal{E}$, then $\vec{x}$ is called \emph{recurrent}. Let $\mathcal{R}$ denote the set of all recurrent states of $(X,D)$. If an ergodic set is made of a single element, such a state is called \emph{absorbing}. Equivalently, $x$ is absorbing when $D_{\vec{x}\vec{x}}=1$. Let $\mathcal{A}$ denote the set of all absorbing states of $(X,D)$. Clearly, an absorbing state is recurrent, hence $\mathcal{A} \subseteq \mathcal{R}$.

The following Proposition proves that there are no recurrent states other than absorbing states, and provides a characterization of absorbing states as maximal states, and thus as myopically team-wise stable states.
    
\begin{proposition}\label{ARMB}
Take a team formation model satisfying Assumption \ref{ass:v0} and a myopic team-wise dynamics. We have that
$\mathcal{A} = \mathcal{R} = \mathcal{M} = MTS$.
\end{proposition}
\begin{proof}    
From the definitions of $\mathcal{R}$ and $\mathcal{A}$, we know that $\mathcal{A} \subseteq \mathcal{R}$. Moreover, from Lemma \ref{lem:myopic=maximal}, we have that $\mathcal{M} = MTS$.

Take $x \notin \mathcal{M}$. An additional project can be formed, and once formed, state $x$ will never be visited again in the future. This shows, by contraposition, that $\mathcal{R} \subseteq \mathcal{M}$ and $\mathcal{A} \subseteq \mathcal{M}$. 

Finally, consider any state $x \in \mathcal{M}$. Since existing projects never disappear \begin{rev}because of Assumption \ref{ass:v0}\end{rev}, and no new project is feasible starting from $x$, we can conclude that $x$ is absorbing. Therefore, $\mathcal{M} \subseteq \mathcal{A}$. 
\end{proof}

\subsection{Perturbed dynamics and stochastically stable states}

We are ready to introduce perturbations in the unperturbed dynamics considered in the previous subsection, and then use the techniques developed by \citet{foyo1}, \citet{youn2}, \citet{kamaro1}. Basically, we suppose that with a tiny amount of probability active projects may accidentally dissolve, and possibly (but not necessarily) non-existing projects can be formed if feasible. By so doing, the dynamic system under consideration becomes ergodic, and from known results it follows that there exists a unique probability distribution among states that is stationary and describes the limiting behavior of the Markov chain as time goes to infinity, irrespectively of the initial state. Then we consider the limit of this stationary distribution for the amount of perturbation decreasing to zero. Those states that are visited with positive probability in this limiting stationary distribution are called stochastically stable.

We invite the reader who is interested in a formal exposition of perturbed Markov chain theory to consult \citet{youn2} and \citet{Ellison00}, while in the following we simply make use of the \emph{resistance} function $r: X \times X \rightarrow \mathbb{R^+} \cup \{\infty\}$, where $r(x,x')$ indicates the minimal amount of perturbations required to move the system from $x$ to $x'$ in one unit of time. If $r(x,x') = 0$ then the system moves from $x$ to $x'$ with positive probability in the unperturbed dynamics, i.e., $T_{xx'}>0$, while $r(x,x') = \infty$ is interpreted as impossibility of moving from $x$ to $x'$ in one unit of time even when perturbations are allowed.

We rely on the techniques and results illustrated in \citet{foyo1}, \citet{youn2} and \citet{youn1}, as they provide a relatively easy way to identify which states are stochastically stable.\footnote{\citet{newton2021conventions}, building upon \citet{peski.10}, provides an asymmetry condition that implies stochastic stability and is applicable to a wide variety of settings.} 
More precisely, we restrict our attention to absorbing states (since there are no other recurrent states by virtue of Proposition \ref{ARMB}), and for any pair $(x,x')$ of absorbing states we define $r^*(x,x')$ as the minimum sum of the resistances between states over any path starting in $x$ and ending in $x'$.\footnote{\begin{rev}It is worth noting that here absorbing states refer to those of the corresponding unperturbed dynamics, $\mathcal A$, while no state is absorbing in the perturbed dynamics.\end{rev}} 
Then, for any absorbing state $x$, we define an $x$-tree as a tree having root at $x$ and all absorbing states as nodes. The resistance of an $x$-tree is defined as the sum of the $r^*$ resistances of its edges. Finally, the \emph{stochastic potential} of $x$ is said to be the minimum resistance over all trees rooted at $x$.

A state $x$ is proven to be stochastically stable \citep{foyo1} if and only if $x$ has minimum stochastic potential in the set of absorbing states. Intuitively, stochastic stability selects those states that are easiest to reach from other states, with ``easiest'' interpreted as requiring the fewest mutations (as measured by the stochastic potential).

We now introduce two alternative perturbation schemes in the unperturbed dynamics considered in the previous subsection. The two perturbation schemes lead us to the same results. The first one is called a \emph{uniform perturbation scheme}, and it is such that at every time $s$ each project $p \in P$ has an i.i.d.~probability $\epsilon$ subject to an error; such an error makes the project disappear if existing, and be formed if non-existing and $x^s \cup \{p\} \in X$. The second perturbation scheme is such that only existing projects can be hit by a perturbation, so that each existing project disappears with an i.i.d.~probability $\epsilon$. We refer to this modeling of errors as a \emph{uniform destructive perturbation scheme}.\footnote{%
For our result we do not require the probabilities to be equal, and every team $\vec{t}$ could have any state and time dependent utility $\eta (p,x,s,\epsilon)$, depending also on some positive real number $\epsilon$. Thus we merely require that all such probabilities converge to zero with the same order as $\epsilon$ goes to $0$.
\label{note:epsilon}
}
It is easy to check that the perturbed system is irreducible and aperiodic: from any state $x \in X$, every existing project may dissolve by means of perturbations, and then one project per period may form, leading the system to any state $x'$. Every project can form in the unperturbed dynamics, thanks to Assumption \ref{ass:v0}, and hence perturbations creating new projects (that are allowed only in the uniform perturbation scheme) do not play any significant role. Aperiodicity is ensured since there are no recurrent states other than absorbing states (see Proposition \ref{ARMB}).

We denote by $SS$ the set of stochastically stable states. The following proposition identifies stochastically stable states as the states with the maximum number of existing projects.

\begin{proposition}\label{prop:stst}
Take a team formation model satisfying Assumption \ref{ass:v0}, a team-wise dynamics and either a uniform destructive perturbation scheme or a uniform perturbation scheme. Then, $SS =\mathcal{L}$\begin{rev}, i.e., the set of states with maximum number of completed projects.\end{rev}
\end{proposition}
\begin{proof} 
Consider any two absorbing states $x, x' \in X$. In order to move from $x$ to $x'$ it is necessary for all projects that exist at $x$ and do not exist at $x'$ to be stopped, and this can only occur by means of $\ell(x) - \ell(x \cap x')$ perturbations, both in the uniform destructive perturbation scheme and in the uniform perturbation scheme. 
In contrast, projects that exist at $x'$ and do not exist at $x$ can form in the unperturbed dynamics, thanks to Assumption \ref{ass:v0}. Therefore: 
\begin{eqnarray}\label{resist}
r^*(x, x') = \ell(x) - \ell(x \cap x').
\end{eqnarray}
We are ready to prove $SS \subseteq \mathcal{L}$. We proceed by contradiction. Suppose $x \notin \mathcal{L}$. Then we can find an $x'$ such that $\ell(x') > \ell(x)$. Take any $x$-tree and consider the path from $x'$ to $x$, say $(x', x_1, \ldots, x_i, \ldots, x_k, x)$. By (\ref{resist}), the sum of resistances over this path is $\ell(x') - \ell(x' \cap x_1) + \ell(x_1) - \ell(x_1 \cap x_2) + \ldots + \ell(x_{k-1}) - \ell(x_{k-1} \cap x_k) + \ell(x_k) - \ell(x_k \cap x)$. We now consider the $x'$-tree obtained from the $x$-tree by reversing the path from $x'$ to $x$. Again by (\ref{resist}), the sum of resistances over this reversed path is $\ell(x) - \ell(x \cap x_k) + \ell(x_k) - \ell(x_k \cap x_{k-1}) + \ldots + \ell(x_2) - \ell(x_2 \cap x_1) + \ell(x_1) - \ell(x_1 \cap x')$. Taking the difference between the above sums of resistances over the two paths, we obtain that the $x'$-tree has a resistance which is equal to the resistance of the $x$-tree $+ \ell(x) - \ell(x')$. Since $\ell(x') > \ell(x)$, we can conclude that for any $x$-tree we can find an $x'$-tree with a lower overall resistance, and hence the stochastic potential of $x'$ is lower than the stochastic potential of $x$. Therefore, $x$ cannot be stochastically stable.

We now prove $\mathcal{L} \subseteq SS$. Since at least one stochastically stable state must exist, and we have just seen that no state outside of $\mathcal{L}$ is stochastically stable, we can therefore conclude that there exists an $x \in \mathcal{L}$ that is stochastically stable. Consider any other $x' \in \mathcal{L}$.
Following exactly the same reasoning as above we obtain that the stochastic potential of $x'$ must be the same as the stochastic potential of $x$. Therefore, $x'$ is stochastically stable as well.
\end{proof}

Proposition \ref{prop:stst} is our main result. It provides a very precise characterization of stochastically stable states for every team formation model that satisfies the general assumption \ref{ass:v0}, and under \begin{rev}both \emph{destructive} and \emph{uniform perturbation schemes}\end{rev} that we have considered.

%\bigskip

To aid intuitive comprehension, we provide the following discussion. In the representation of all the states in $X$ as a partially ordered set, with the empty state at the top and maximal states at the bottom, an error can be seen as a step upwards, while the adapting best response of agents (under Assumption \ref{ass:v0}) can be seen as a step downwards.
Consider Example \ref{ex:POset} as represented in Figure \ref{fig:POset}.
In this case, the bottom state is $SS$, because it would need at least two errors before the system can move in the unperturbed dynamics to another $MTS$ state.
To move away from the other two $MTS$ states, on the other hand, only one error is required.

%\bigskip

\begin{rev}
In line with one of the main results by \cite{JacksonWatts02}, which is that stochastic stability selects the complete network among the many network structures obtained as pairwise equilibria, in our framework this is obtained directly by applying Proposition \ref{ARMB}: if the technology allows it, the complete network is the one that maximizes the number of projects realized. %\ref{prop:stst}
\end{rev}

Clearly, the fact that $SS$ states maximize the number of teams in the state does not tell us anything about efficiency, because the utility function could have a structure that highly rewards agents in states with few teams.
\begin{rev}
Actually, in the APS dataset discussed in Section \ref{subsec:empirical_motivation} the evidence seems to suggest that researchers tend to accept every new project, but this causes inefficiencies because of negative externalities and congestion. However, under Assumption \ref{ass:v1}, we can prove the following corollary.\end{rev}

\begin{corollary}
Given a team formation model, under Assumption \ref{ass:v1}, every stochastically stable state is Pareto efficient.
\end{corollary}
\begin{proof}
By Proposition \ref{prop:stst} a stochastically stable state $x$ maximizes the number of projects $\ell (x)$.
By Assumption \ref{ass:v1} we have $ \ell (x) = \sum_{i \in N} u_i ( x  ) $, so that $x$ also maximizes the aggregate utility to agents.
Thus, there is no other state that can provide a higher utility for some agent, without damaging any other agent.
\end{proof}

%Note here that also Assumption \ref{ass:v1} is very general, as we are speaking about value functions, whereas in principle agents could have any kind of even heterogeneous non-satiated utility functions upon those values.
%Moreover, we allow for any kind of externality in those values.

\subsection{Descriptive value of stochastic stability}\label{section:descriptive}

In this section we have shown that in a team formation model stochastically stable states coincide with \begin{rev}the maximal states with the maximum number of projects\end{rev}. 
This result rests not only on Assumption \ref{ass:v0} and a large class of perturbation schemes, but also on individual behavior that is boundedly rational. In particular, no agent will ever exit from existing projects in order to free up time and start other projects, despite the fact that this may increase her utility. 
Given this circumstance, one may then query the descriptive value of stochastic stability. An answer clearly depends on the specific case under consideration. We limit ourselves to the following observations. Even if agents have sufficient cognitive skills to recognize the possibility of an increase in utility, there are at least two kinds of reasons that might prevent them from doing so. First, coordination issues: in order to carry out a utility enhancement, an agent has to quit projects with some teams and contextually start other projects with new teams, and this involves coordinating the actions of several agents. Second, switching costs: these costs may be due to legal obligations -- consider for instance divorce costs in marriage -- or to learning how to operate in new teams. 

Nevertheless, in the following section we analyze a refinement of myopic team-wise stability based on strong rationality and absence of coordination and switching costs. 
%[We will show that such a refinement -- a part from its questionable descriptive value -- can be undesirable also from a normative point of view.] \textbf{[The previous sentence: to be removed?]}

In \ref{app:robust} we consider a variant of the team formation model where agents are endowed with unlimited cognitive and coordination skills, but
face switching costs when leaving an existing project. We show that when switching costs are high enough, stochastically stable states are exactly those states having the maximum number of existing projects.

\section{Coalitional stability}\label{section:cs}

In this section we introduce a concept of stability -- \emph{coalitional stability} -- which is strongly based on coordination opportunities and rationality of agents. Then we compare it with myopic team-wise stability and stochastic stability, providing a class of situations where coalitional stability has no refining power.

\subsection{Coalitionally stable states}\label{section:def_cs}

We consider the following definition:
\begin{definition}
\label{def:cs}
A state $x$ is \emph{coalitionally stable} {\bf [CS]} (or \emph{coalition proof}) if there exists no subset $C \subseteq N$ such that
\begin{itemize}
%\item[(i)] there exists a state $y \subseteq x$ such that $\forall ~ p  \in y$, $\exists ~ i \in C$ such that $i \in n(p)$;
\item[(i)] there exists a set of projects $y \subseteq x$ such that $\forall ~ p  \in y$, $\exists ~ i \in C$ such that $i \in n(p)$;
%\item[(ii)] there exists a state $z$ such that $\forall ~  p  \in z$, if $j \in n(p)$ then $j \in C$;
\item[(ii)] there exists a set of projects $z$ such that $z \cap x = \emptyset$ and, $\forall ~  p  \in z$, if $j \in n(p)$ then $j \in C$;
\item[(iii)] $ (x \backslash y ) \cup z \in X$ and for any agent $i \in C$ we have that $u_i ( (x \backslash y ) \cup z ) > u_i (x)$. 
\eproof
\end{itemize}
\end{definition}

In words, a state $x$ is \emph{coalitionally stable} if there is no coalition that can $(i)$ erase a set of projects, %(i.e., a state, by definition),
each of which contains at least one agent of the coalition, $(ii)$ form a set of other projects, where all the members of each are also members of the coalition, and $(iii)$ make all the members of the coalition strictly better off in the resulting state.\footnote{In contrast to what happens for myopic team-wise stability (see footnote \ref{foot2}), the choice to require a strict Pareto improvement for the agents of a blocking coalition -- instead of asking for a weak Pareto improvement -- can enlarge the set of coalitionally stable states even when Assumption \ref{ass:v0} is satisfied. However, this difference ceases to exist when we introduce costs to exit from existing projects (see \ref{app:robust}).}
We denote by $CS$ the set of states that are coalitionally stable.

It is evident that this definition allows for a profound rationality by the agents: they can identify and coordinate a deviation through a long path in the partially ordered set $X$.
\begin{rev}
Notice, in particular, that in $CS$ deviations by larger coalitions are allowed, so the equilibrium requirement is more binding than under  $MTS$, hence $CS \subseteq MTS$. Whether this inclusion can become an equality is discussed more in depth in Proposition \ref{prop:coalitional=team-wise}.\end{rev}\footnote{\begin{rev}The concept of coalitional stability is analogous to the \emph{core} in cooperative game theory, in a context of coalition formation. To define the core one needs to explicitly define the agents' strategies, which is possible but not useful for our analysis.\end{rev}}
This definition \begin{rev}of $CS$\end{rev} allows agents to maintain some of their existing teams with other agents outside of the coalition, and in this sense it is a generalization of \emph{bilateral deviations} defined in a network formation setting by \citet{GV07}.
The literature on clubs (e.g., \citealp{Pauly70} and \citealp{faias2017endogenous}) focuses on deviations where all \emph{clubs} are deleted when members are both outside and inside the coalition.
Note finally that we are not providing a general result of existence of $CS$ states, for which we may need a more general concept of stability such as the one provided in more specific settings by \citet{HMV09,HMV10} and \citet{mauleon2018constitutions}.
However, we will focus on situations in which we can compare $CS$ states with $SS$ states, so that a simpler definition suffices.\footnote{%
Some clarification is needed.
The definition of \emph{farsighted coalitionally stable} states, as proposed by \citet{HMV09,HMV10} in a context where agents are farsighted players who evaluate the desirability of a deviation in terms of its future consequences (see also \citealp{dutta2005farsighted} and \citealp{navarro2013expected}), can easily be generalized to our context, and this is what we achieve in \ref{App:FSS}.
The good thing about the above definition is that there always exists a farsighted coalitionally stable state, while existence is not guaranteed for simple coalitionally stable states, as we define them.
Accordingly, they are a super-set of the \emph{coalitionally stable} states.
However, our point is that in many contexts coalitional stability has too little predictive power, so that the coalitionally stable states are too numerous or can be anything.
We are not concerned here with the point that in other contexts coalitional stability can be a concept so restrictive that no state satisfies it.}
In \ref{app:robust} we provide more general definitions, which integrate this approach with that of stochastic stability.

\subsection{A comparison between notions of stability}\label{section:comp}

In the representation of all the states in $X$ as a partially ordered set, the deviation of a coalition can be seen as a path that moves first upwards and then downwards. 
Consider as an example the whole states of Example \ref{ex:POset} and Figure \ref{fig:POset}.
Suppose that in this case the utility that agents $j$ and $k$ receive from being together is always greater than the utility they receive from being respectively with $i$ and $m$.
Then, the bottom state that maximizes the number of projects would not be $CS$:
$j$ and $k$ could coordinate to delete all the existing projects and start two projects together, moving to the state on the extreme right.
Therefore, in general, $SS$ and $CS$ states are not necessarily related concepts and can have empty intersections.
This can also clearly be seen if we look back at Example \ref{example:indirect_externalities} and the related Figure \ref{fig:burgerking_s}.
$I$ and $II$ are stochastically stable, because they maximize the number of teams at four, while $II$ and $III$ are coalitionally stable, because $I$ can be broken by the coalition $\{ j,k\}$.

There are cases in which $SS$ and $CS$ states can be a subset of one another, and given the flexibility of the utility function $\vec{u}$, a setup can always be provided, even under Assumption \ref{ass:v0}, such that any desired subset of states would always be chosen by the grand coalition $N$. 

On the other hand, there are many other cases, based on simple and general utility functions, where the $CS$ states do not provide a clear or apparently improving refinement upon the $MTS$ states.
As an example, consider that in general $CS$ states may not even be Pareto efficient, because agents adhere to deviating coalitions only if their marginal profit from doing so is strictly positive: thus, it could be that a Pareto improving deviation is feasible through a coalition whose members will not all be strictly better off. 
An instance of this sort is in the case provided by Example \ref{ex:POset}, under Assumption \ref{ass:v2}: the top right maximal state of Figure \ref{fig:POset}, with only agents $b$ and $c$ forming teams, is a $CS$ state, even if it is Pareto dominated by the bottom left state (with all agents taking part in two teams each).

Here below we present a result for a case (which actually encompasses the previous example) where $CS$ states coincide with $MTS$ states, so that they provide no refinement at all with respect to the myopic and boundedly rational concept of myopic team-wise stability.

\begin{proposition} \label{prop:coalitional=team-wise}
Given a team formation model, under Assumptions \ref{ass:t1} and \ref{ass:v2}, we have that $CS = MTS$.
\end{proposition}

\begin{proof}
Since \begin{rev} larger coalitions are allowed to deviate under $CS$\end{rev}, we always have that $CS \subseteq MTS$. We focus on showing that $MTS \subseteq CS$.
Let us suppose this is not the case; then there is a team-wise stable state $x$ which is not coalitionally stable.
From the definition, this means that in $x$ there is a coalition $C$ that can erase a set $y$ of projects and start a set $z$ of new projects.

As agents in $C$ need to strictly increase their utility (i.e., the number of teams to which they belong, from Assumption \ref{ass:v2}), for each of the agents there is at least one unit of free time in $x$; formally we have that $e_i (x) < w_i$ for each $i \in C$.

As all the projects give strictly positive payoffs, then $z$ is nonempty.

But then, as each team in $z$ is formed by all and only agents from $C$, and by Assumption \ref{ass:t1} it would cost one unit of time each, then each project $p \in w$ could be started already in $x$, and $x \cup \{p \} \in X $, i.e., it would be feasible.
We have reached a contradiction with the hypothesis that $x$ is a team-wise stable state.
\end{proof}

\section{Conclusions and future research}\label{section:conclusion}

In this paper we have provided a model which describes how teams of individuals arise in order to perform activities, investing amounts of a scarce resource (typically time) to conduct the activities. The kinds of interaction that can be modeled in the proposed framework are many and widespread in economic and social spheres. Unfortunately, in a context like this the complexity of analysis can increase very rapidly, so that predictions become very hard to make. Nevertheless, we introduce and discuss alternative notions of stability -- myopic team-wise stability, stochastic stability and coalitional stability -- and we are able to provide results that are rather clear-cut (especially for stochastic stability).

Future work can highlight the relevance of the model for specific applications. The setting that we have provided is sufficiently general and flexible to accommodate many different sets of assumptions, and this allows a proper fine-tuning of the model. The examples throughout the paper  illustrate its applicative potential. In \ref{app:robust}, a variant of the model is presented where we add switching costs, which can be considered as a realistic feature for several applications. A promising direction for research could explore its potential applicability to the job market, where the result that stochastic stability selects states with the highest number of projects has an interesting interpretation in terms of unemployment reduction. 

On a purely theoretical ground, we provide here below three possible lines along which research may lead to interesting advancements.

The first question is related to the concept of coalitional stability that we define in Section \ref{section:cs}.
We are able to provide examples of non-existence of coalitionally stable states, but also (as in Proposition \ref{prop:coalitional=team-wise}) to prove their existence under specific assumptions.
We partly address this question in \ref{App:FSS}, where we define a wider set of \emph{farsightedly stable states}, whose existence is always granted.
However, we conjecture that the existence of simple coalitionally stable states holds also for more general assumptions than those of Proposition \ref{prop:coalitional=team-wise}.

The second question is related to welfare issues, and follows from the discussion in Sections \ref{section:ss} and \ref{section:cs}: under which conditions on the utility function are $SS$ states and $CS$ states Pareto efficient, or do they maximize the objective function of some social planner?
Clearly a simple case is the one in which the objective function is monotonic in the number of teams, so that $SS$ states are \emph{optimal}; or the case in which the utility function is monotonic for all agents in the objective function, so that $CS$ states would be \emph{optimal}, because even the grand coalition made of all the $N$ agents is better off in those states.
But how much of the two simple statements above can be generalized in order to have non-trivial results?

The third question follows from the previous one but has a mechanism design approach. 
Suppose that, through incentives, a planner can slightly change the utility function, not to obtain the trivial forms discussed above but something approaching that. Alternatively, the planner could have the possibility of modifying the technology, at least to the point at which feasibility depends on the structure of connections among agents: only agents that are close (in some sense) can work together, and the planner can adopt policies to affect who is close to whom.
What are the sufficient conditions that would allow the planner to make $SS$ or $CS$ states Pareto efficient, or make them approach maxima of some objective function?

\begin{rev}
The last two questions could also have empirical applications and, possibly, policy implications if applied to contexts like academic research and organization of researchers' communities, as exemplified by our motivational description of the APS dataset in Section \ref{subsec:empirical_motivation}. In such environments, there are many possible objective functions that a social planner can aim at, and choosing one would be the first step to start thinking about an ideal mechanism.
\end{rev}

% \begin{rev}The last two questions could have also empirical applications, if applied to contexts like the APS dataset that we use as motivation for our approach in Section \ref{subsec:empirical_motivation}. In such contexts there are many possible objective function that a social planner can aim at, and choosing one would be the first step to start thinking about an ideal mechanism.
% \end{rev}

% \section*{Acknowledgements}
% \noindent We thank the following for helpful comments: Andrea Galeotti, Sanjev Goyal, Matthew Jackson, Shachar Kariv, Brian Rogers, Fernando Vega Redondo, Simon Weidenholzer, Leeat Yariv and seminar participants at the ASSET 2012 Meeting in Cyprus, Berkeley, Caltech, Essex University, European University Institute, LUISS University in Rome, University of Siena and Stanford University.
% We acknowledge funding from the Italian Ministry of Education Progetti di Rilevante Interesse Nazionale"
% (PRIN) grants 2015592CTH and 2017ELHNNJ.

%P.~P.~received support from the Italian Ministry of Education, Universities and Research under FIRB project RBFR1269HZ ``Social and spatial interactions in the accumulation of civic and human capital''.

\bibliographystyle{chicago}
\bibliography{ss_mt}

\setcounter{section}{1}

\appendix
%dummy comment inserted by tex2lyx to ensure that this paragraph is not empty
\global\long\def\thesection{Appendix \Alph{section}}
 \global\long\def\thesubsection{\Alph{section}.\arabic{subsection}}
 \setcounter{equation}{0} \global\long\def\theequation{\alph{equation}}
 \setcounter{proposition}{0} \global\long\def\theproposition{\Alph{proposition}}
 \setcounter{definition}{0} \global\long\def\thedefinition{\Alph{definition}}

\newpage
\begin{new_am}
\section{Additional information on APS data}
\label{app:data}

In this section we provide some figures that help to describe the data analyzed in Section \ref{subsec:empirical_motivation}.
To analyze individual careers, we focus on authors that have a lengthy and consistent career of at least 25 years and whose first year of publication, called author's \emph{cohort}, is after 1960. 
Figure \ref{fig:app1} contains information about this sub-sample of authors and, in particular, it shows the distributions of the number of authors per cohort (left), the authors' career lengths (right) and the distribution of authors' \emph{active years}, where an year is considered active year for an author if she has published at least one paper in that year. Consequently, the number of an author's active years is the count of such years for that author.
Such Figure then shows that this subsample consists of a set of authors that are well spread terms of cohort year (the first cohorts, 1960-65, obviously are less numerous but afterwards the number of authors in each cohort remains constant enough), that the median author has a career of 32 years and has 15 active years (meaning one papers recorded every 2 years).

\begin{figure}
    \centering
    \includegraphics[width=.7\textwidth]{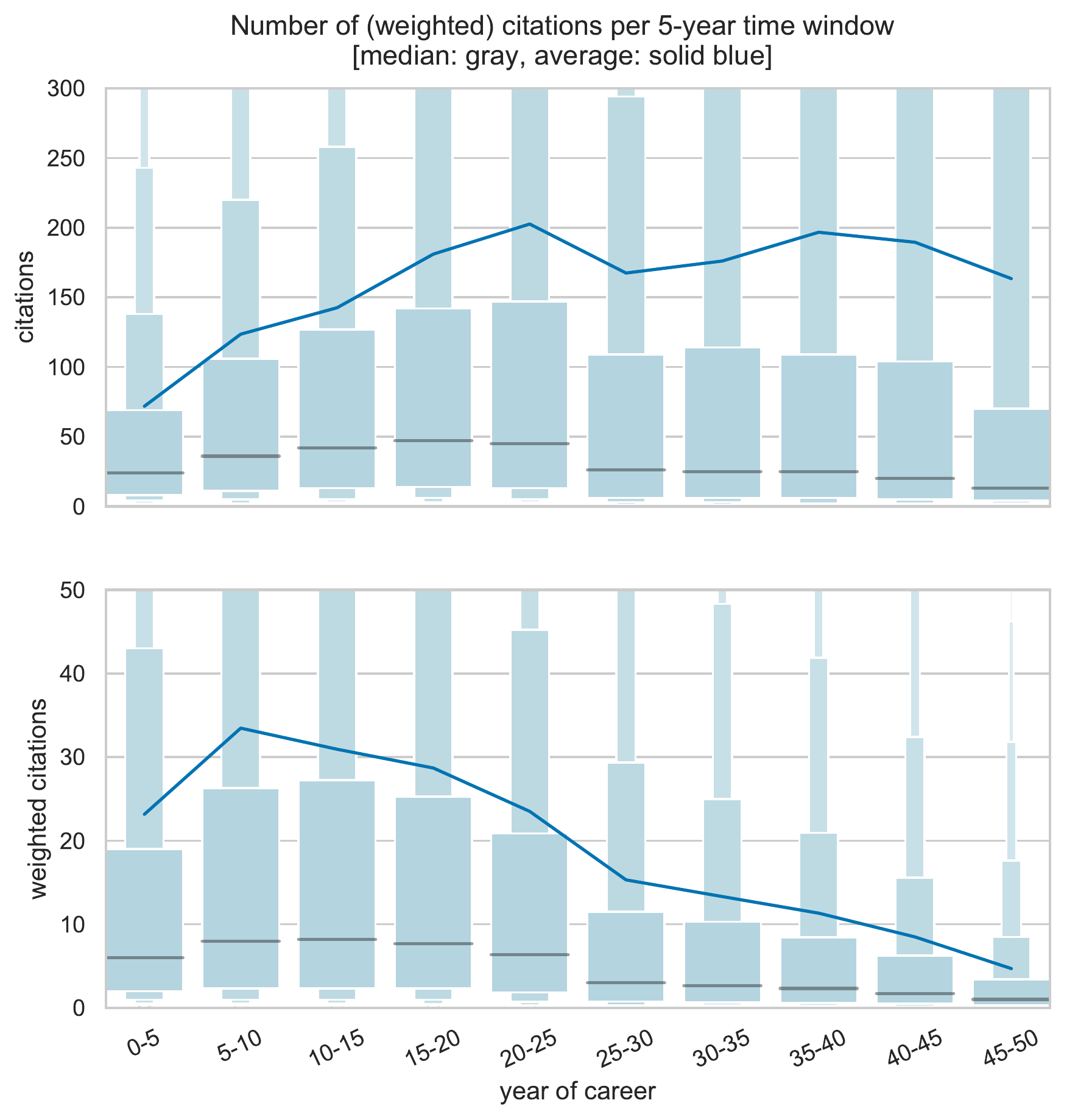}
    \caption{Trends of citations to individuals. A weighted citation received by a paper is computed as $\frac{1}{\text{n authors of receiving paper}}$.}
    \label{fig:2}
\end{figure}

\begin{figure}[ht]
    \centering
    \includegraphics[width=.55\textwidth]{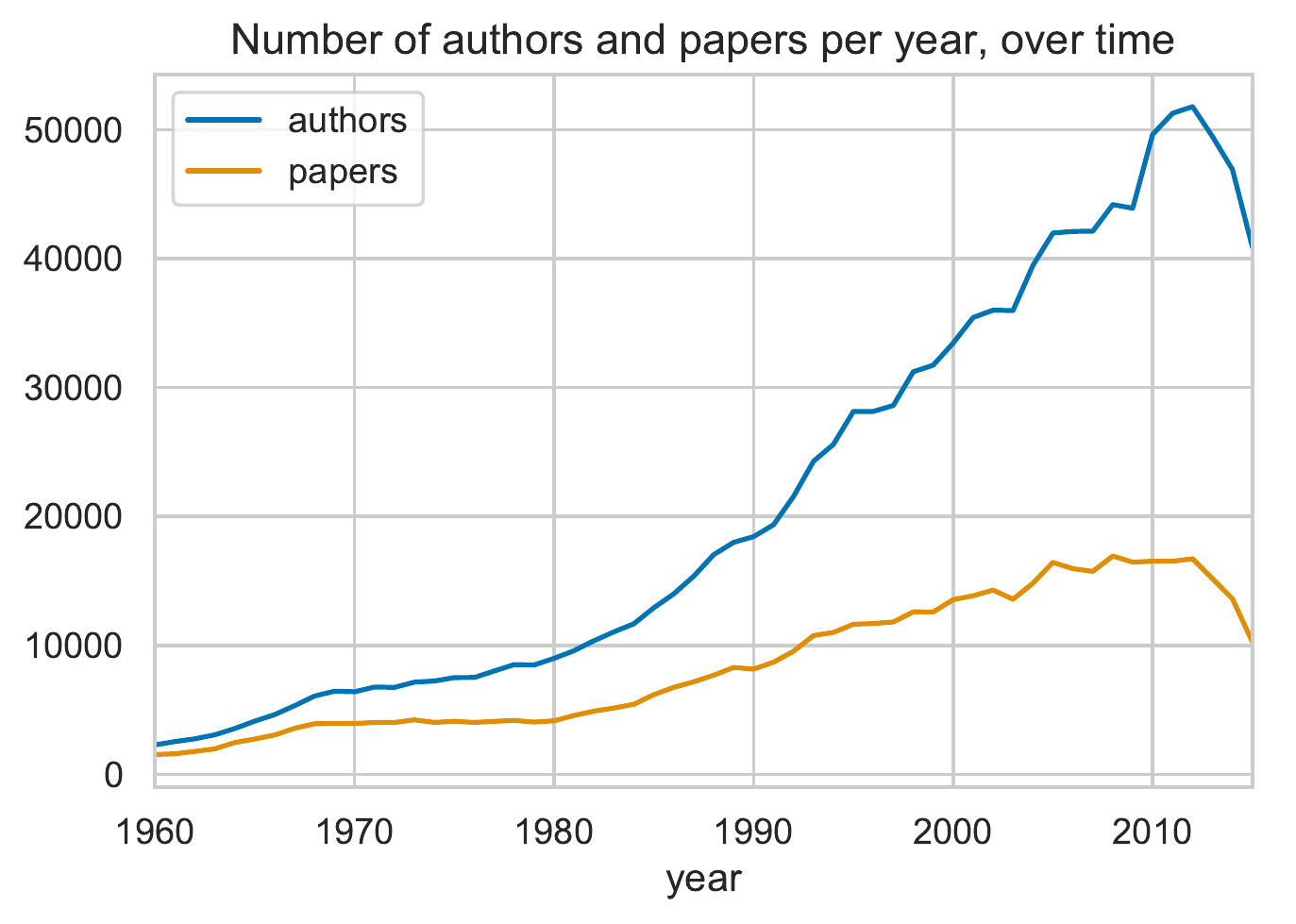}
    \caption{}
    \label{fig:3}
\end{figure}

\begin{figure}[h]
    \centering
	\begin{subfigure}[t]{.32\textwidth}
		\includegraphics[width=\textwidth]{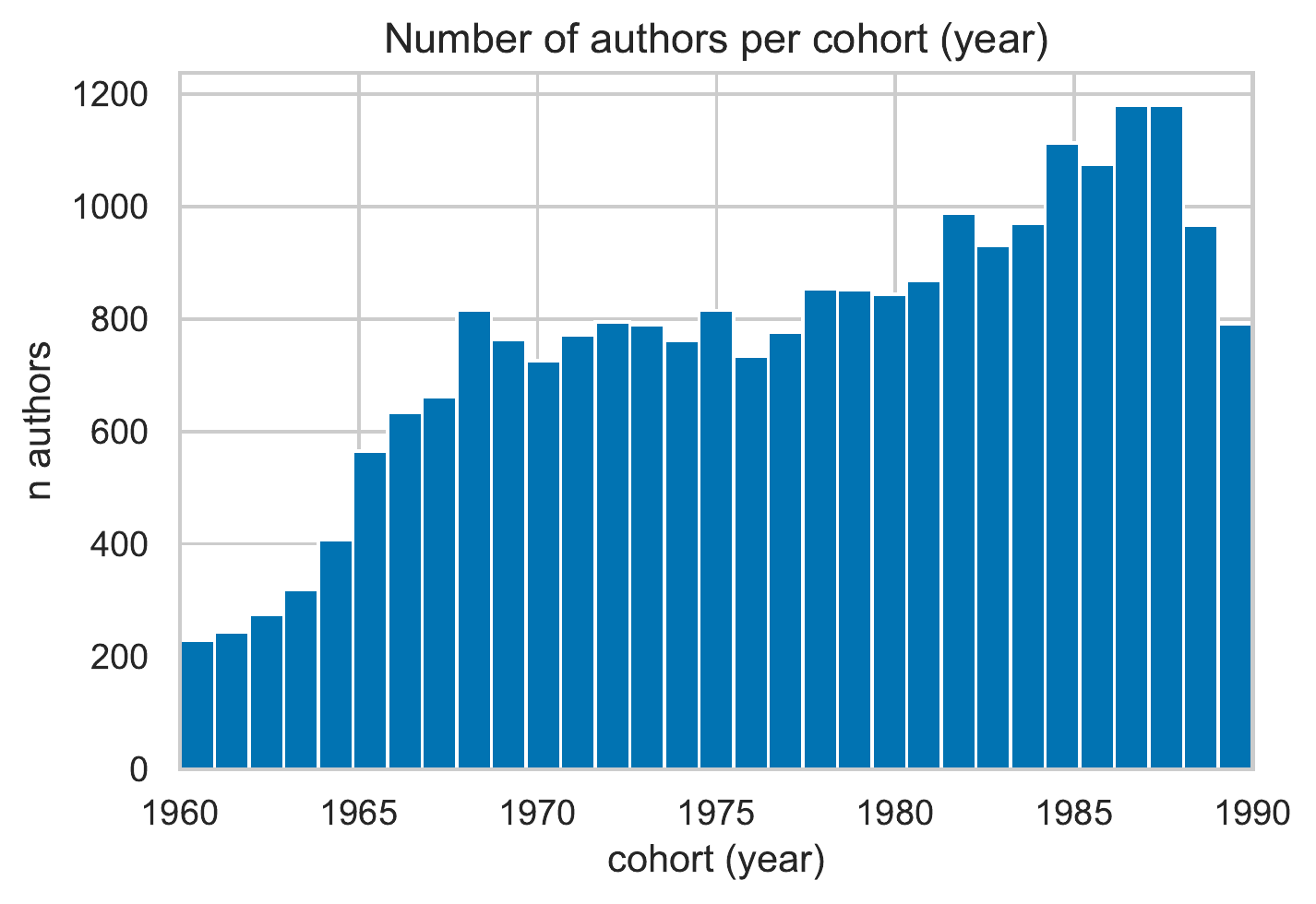}
	\end{subfigure}
	\begin{subfigure}[t]{.32\textwidth}
		\includegraphics[width=\textwidth]{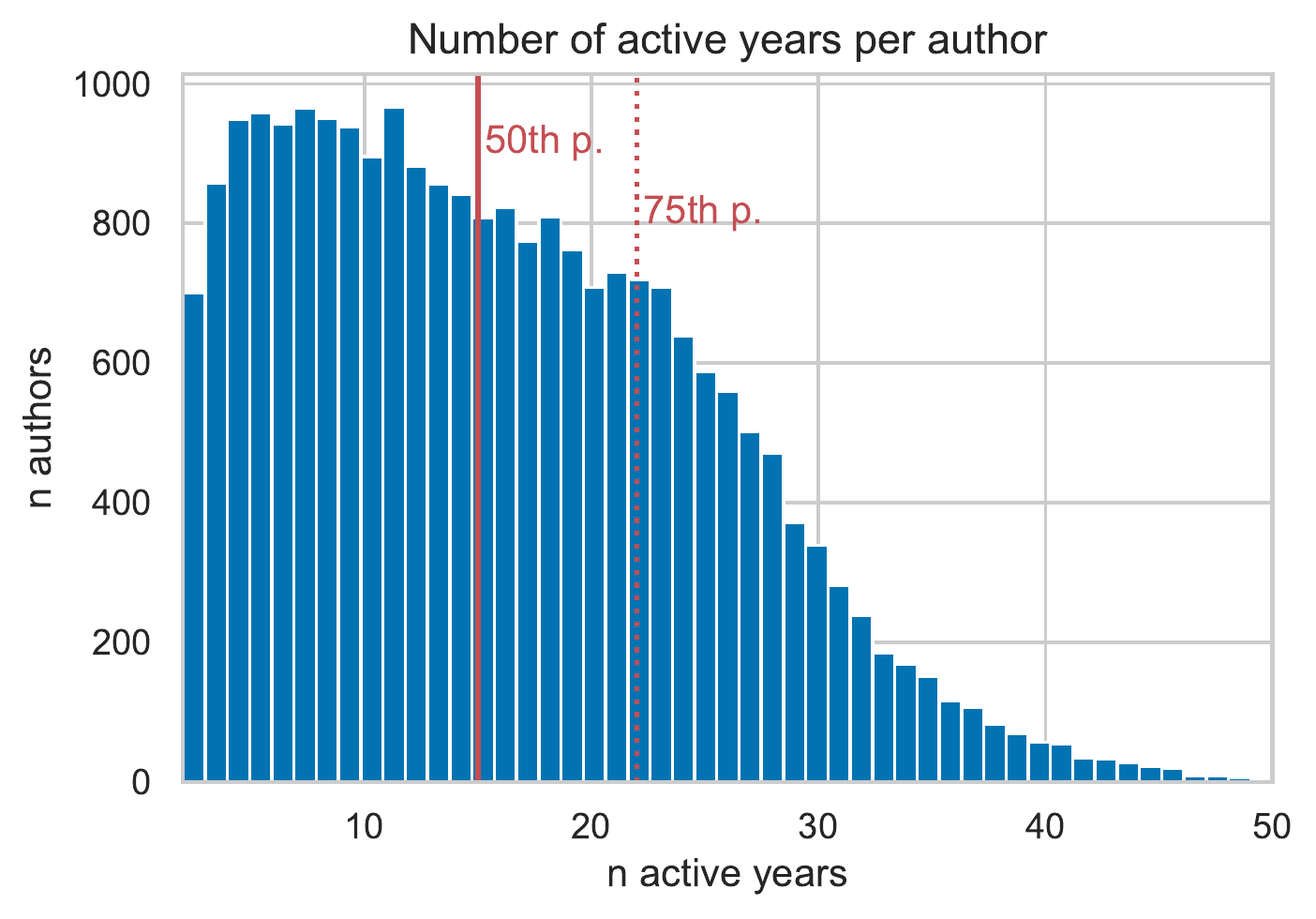}
	\end{subfigure}
	\begin{subfigure}[t]{.32\textwidth}
		\includegraphics[width=\textwidth]{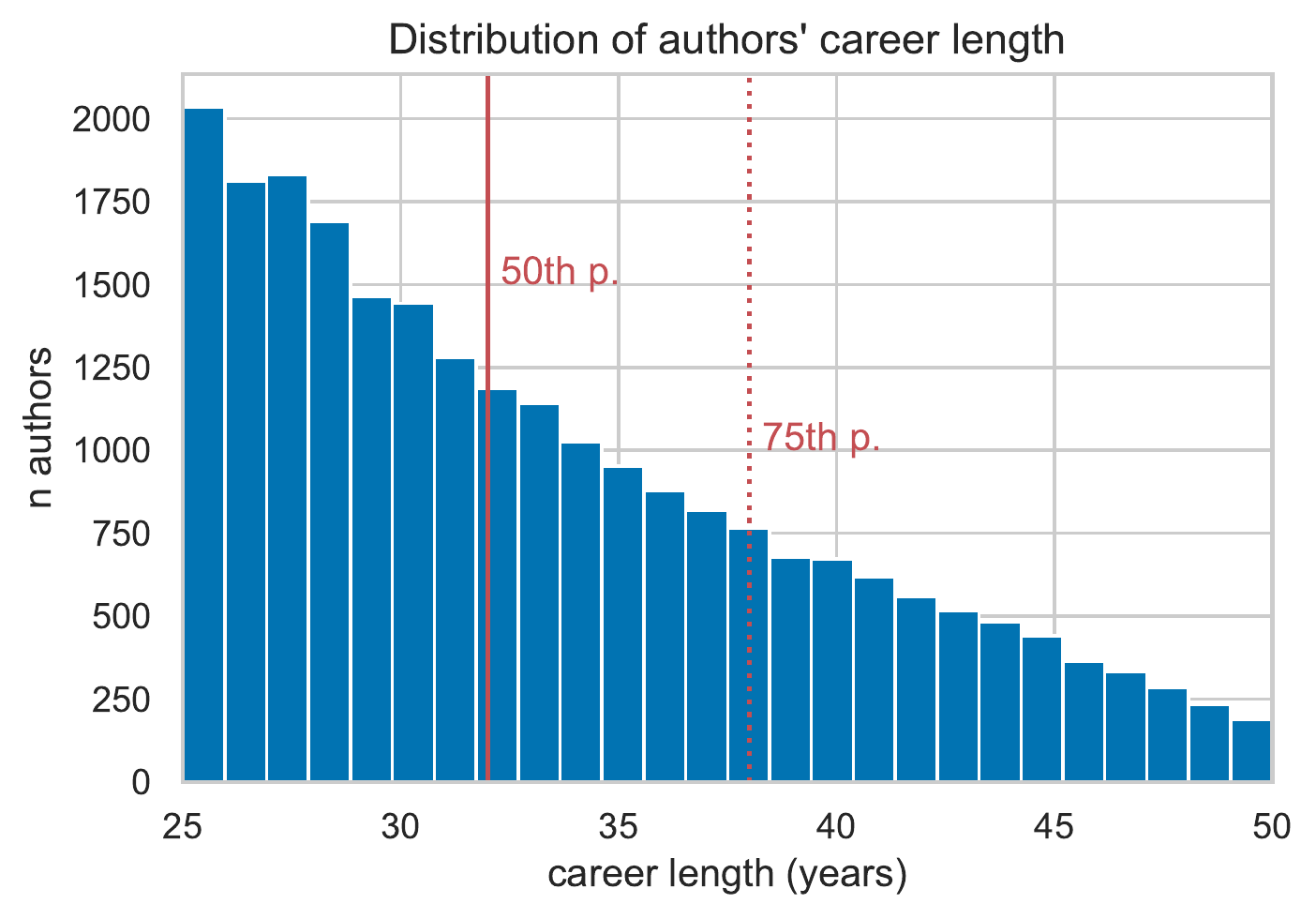}
	\end{subfigure}
	\caption{Descriptive statistics for authors in subsample ($\sim$ 24,000 authors): authors whose cohort appeared after 1960 and whose career length is 25 to 50 years.}
	\label{fig:app1}
\end{figure}

Figure \ref{fig:app2} contains the distribution of citations received papers (left), which, with the help of a logarithmic scale on the y-axis, shows that while most papers (i.e., 75th percentile) receive less than 20 total citations, the top 1 percentile of papers receive more than 100 citations. In Figure \ref{fig:app2} (right), the distribution of papers' citation age (that is, how many years have passed from the citing paper to the cited paper) shows that most of a paper's citations accumulate in the first 5-7 years after its publication.

\begin{figure}[tb]
    \centering
	\begin{subfigure}[t]{.49\textwidth}
		\includegraphics[width=\textwidth]{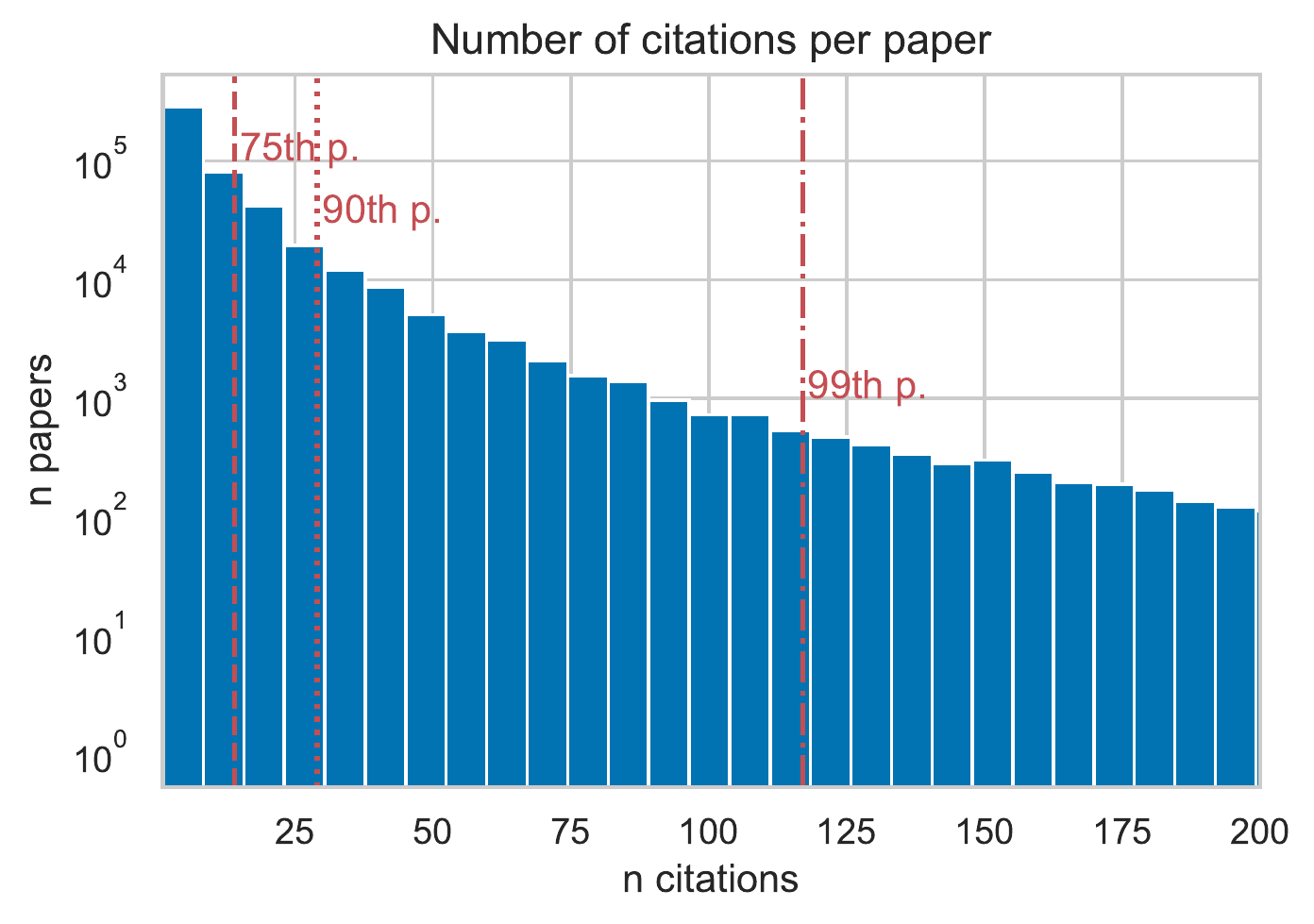}
	\end{subfigure}
	\begin{subfigure}[t]{.49\textwidth}
		\includegraphics[width=\textwidth]{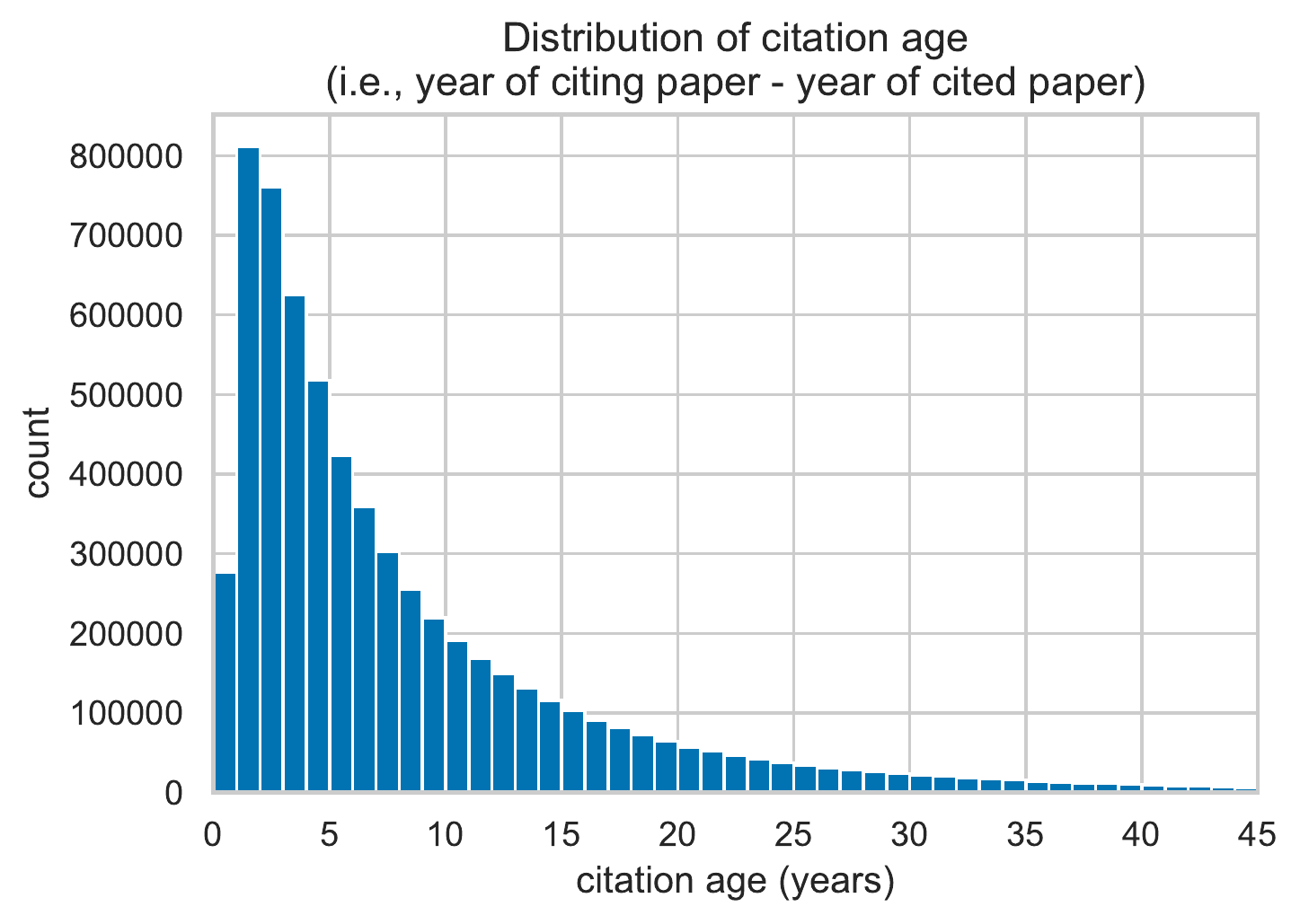}
	\end{subfigure}
	\caption{On the left, the well-known fat tails exhibited by papers' citations. On the right, the distribution of the years occurred from publication of the cited paper to the publication of the citing paper.}
	\label{fig:app2}
\end{figure}

Lastly, in Figure \ref{fig:app3} (left) we plot the increasing trends over the years of the fraction of papers that receive few citations (that is, less than 5 citations). While in Figure \ref{fig:app3} (right) we show that, over time, the number of isolated authors have diminished and also that the number of connected components has remained stable -- or has slightly decreased -- even if the number of authors participating to the profession, that is, number of nodes in the co-authorship network, has increased drastically (see Figure \ref{fig:3}).

\begin{figure}[tb]
    \centering
	\begin{subfigure}[t]{.49\textwidth}
		\includegraphics[width=\textwidth]{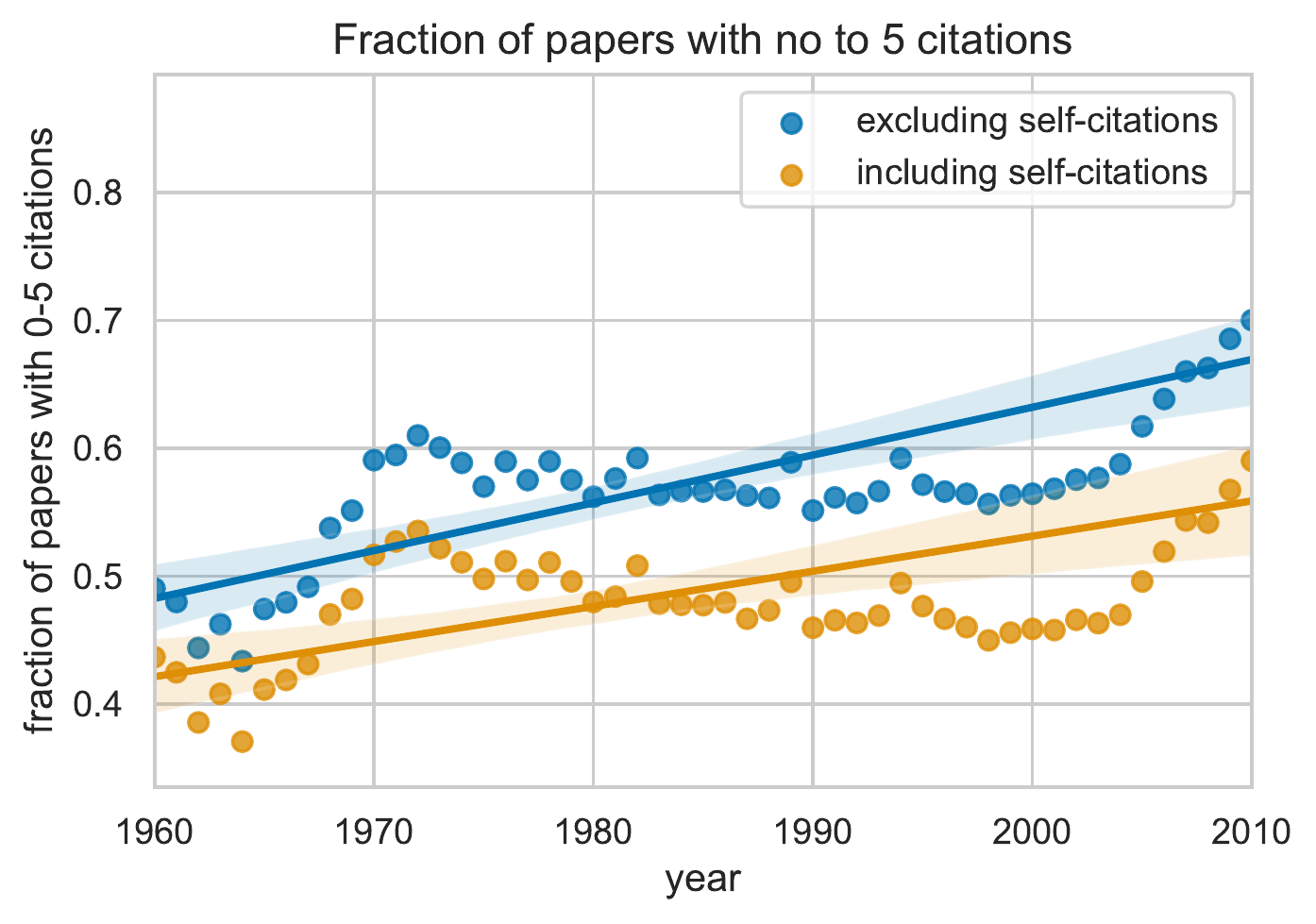}
	\end{subfigure}
	\begin{subfigure}[t]{.49\textwidth}
		\includegraphics[width=\textwidth]{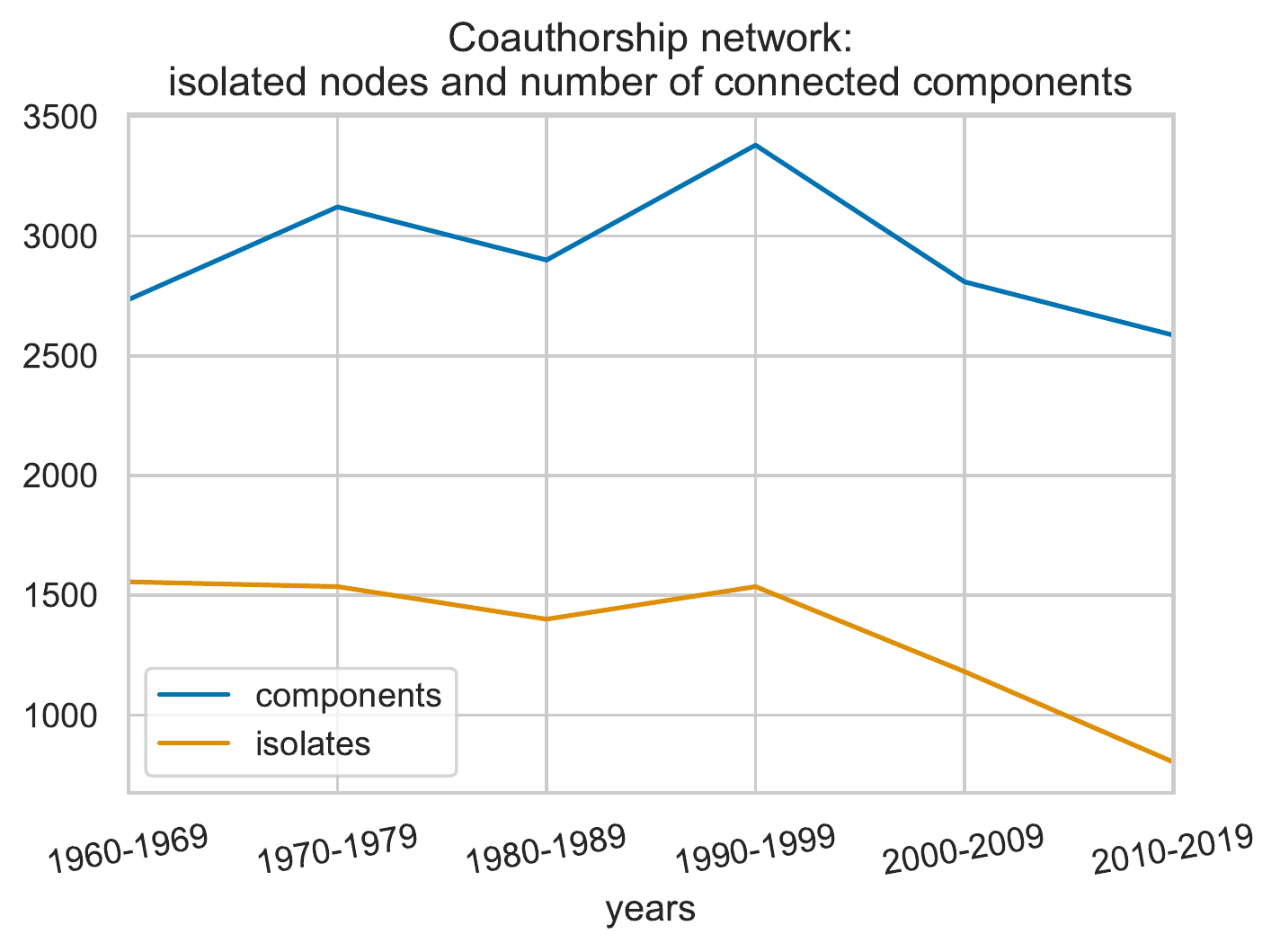}
	\end{subfigure}
	\caption{On the left, the fraction of paper with less than 5 citations. On the right, the number of connected components and of isolated nodes in the coauthorship networks corresponding to each decade.}
	\label{fig:app3}
\end{figure}
\end{new_am}

\clearpage

\section{Robustness of stochastic stability with respect to switching costs}
\label{app:robust}
%[FOR ONLINE PUBLICATION]\\
In this appendix we endow agents with the same degree of rationality and coordinating abilities that we have used in Section \ref{section:cs} for coalitional stability. However, we introduce switching costs to exit from existing projects, showing that if such costs are sufficiently high then coalitional stability coincides with myopic team-wise stability. Moreover, we consider an unperturbed dynamics based on coalitional stability, and we reproduce the results of Proposition \ref{ARMB}. Finally, we show that the introduction of uniform perturbations like in Section \ref{section:ss} yields the same predictions as Proposition \ref{prop:stst}, i.e., stochastically stable states are maximal states.  

The following definition provides an adjusted notion of coalitional stability.

\begin{definition}
A state $x$ is \emph{coalitionally stable} {\bf [CS]} (or \emph{coalition proof}) if there exists no subset $C \subseteq N$ such that
\begin{itemize}
%\item[(i)] there exists a set of projects $y$ such that $\forall ~ p \in y$ we have that $p \in x$ and $\exists ~ i \in C$ with $i \in n(p)$;
\item[(i)] there exists a set of projects $y \subseteq x$ such that $\forall ~ p  \in y$, $\exists ~ i \in C$ such that $i \in n(p)$;
%\item[(ii)] there exists a set of projects $z$ such that $\forall ~ p \in z$ we have that $p \notin x$ and $i \in n(p)$ $\forall ~ i \in C$;
\item[(ii)] there exists a set of projects $z$ such that $z \cap x = \emptyset$ and, $\forall ~  p  \in z$, if $j \in n(p)$ then $j \in C$;
%\item[(iii)] $ (x \setminus y) \cup z \in X$ and for $\forall ~ i \in C$ we have that $u_i ((x \setminus y) \cup z) - c  \sum_{p \in y} 1_i(n(p)) > u_i (x)$. 
\item[(iii)] $ (x \backslash y ) \cup z \in X$ and for any agent $i \in C$ we have that $u_i ((x \setminus y) \cup z) - c  \sum_{p \in y} 1_i(n(p)) > u_i (x)$. 
\eproof
\end{itemize}
\end{definition}
The only difference with respect to the definition provided in Section \ref{section:cs} concerns point (iii), in particular the presence of switching costs that affect the advantage of exiting from projects and forming new ones. More precisely, each agent is aware that she will pay a cost equal to $c \geq 0$ for each project she wants to leave. Function $1_i(n(p))$ denotes the indicator function, giving $1$ if agent $i$ belongs to project $p$, and $0$ otherwise. We denote by $CS(c)$ the set of coalitionally stable states when switching costs are equal to $c$.

For ease of presentation of our argument, we define a blocking operation as a quintuple $(x,C,y,z,c)$ such that for every member $i$ of coalition $C$ we have that $u_i ((x \setminus y) \cup z) - c  \sum_{p \in y} 1_i(n(p)) > u_i(x)$. We denote by $BO(c)$ the set containing all blocking operations when switching costs are $c$.

The following proposition shows a kind of continuity of coalitional stability with respect to $c$. In the presence of a tiny amount of switching costs, the predictions given by coalitional stability do not change compared to the case without switching costs. By contrast, when $c$ is large enough, then coalitional stability gives the same predictions as myopic team-wise stability. In the proposition we use $MTS$ and $CS$ as defined in Sections \ref{section:mts} and \ref{section:cs}.

\begin{proposition}\label{prop:swco1}
Take a team formation model satisfying Assumption \ref{ass:v0}. Then:
\begin{itemize}
\item[(i)] if $c$ is low enough, then $CS(c) = CS$;
\item[(ii)] if $c$ is high enough, then $CS(c) = MTS$.
\end{itemize}
\end{proposition}

\begin{proof} 
Preliminarily, we provide the following thresholds for switching costs:
\begin{eqnarray*}
\underline{c} &=& \min_{(x,C,y,z,0) \in BO(0)} \min_{i \in C} u_i ((x \setminus y) \cup z) - u_i(x)\\
\overline{c} &=& \max_{(x,C,y,z,0) \in BO(0)} \min_{i \in C} u_i ((x \setminus y) \cup z) - u_i(x)
\end{eqnarray*}
We observe that $\underline{c}$ and $\overline{c}$ are well-defined, since a maximum and a minimum always exist in finite sets, and $N$, $X$ and $P$ are all finite.

To prove point (i), we show the double implication. We start from $CS \subseteq CS(c)$ for every $c \geq 0$. Take a state $x \in CS$. We know that no blocking coalition exists in the absence of switching costs, and no blocking coalition can arise when switching costs are added, since the advantage for a coalition to exit from projects and form new ones cannot increase. Hence, $x \in CS(c)$ for every $c \geq 0$. 

We now show that $CS \subseteq CS(c)$ if $c$ is sufficiently low. We fix $c < \underline{c}$ and we consider a state $x \in CS(c)$. For such a state no blocking operation exists when switching costs are $c$. Suppose ad absurdum that a blocking operation exists in the absence of switching costs, i.e, there exists $(x,C,y,z,0) \in BO(0)$. This means that there must exist $i \in C$ such that $u_i ((x \setminus y) \cup z) > u_i(x)$ and $u_i ((x \setminus y) c \sum_{p \in y} 1_i(n(p)) \cup z) > u_i(x)$, but this is in contradiction with $c < \underline{c}$. Hence, $x \in CS$ if $c < \underline{c}$.

To prove point (ii), we exploit Lemma \ref{lem:myopic=maximal}, and we show the double implication between $CS(c)$ and $\mathcal{M}$. Clearly, $CS(c) \subseteq \mathcal{M}$ for every $c \geq 0$. By contraposition, if $x \notin \mathcal{M}$, then there exists a project $p \in P$ such that $p \notin {x}$ and $x \cup \{p\} \in X$. Due to Assumption \ref{ass:v0}, $(x,n(p),\emptyset,\{p\},0) \in BO(0)$, and hence $x \notin CS$. 

Finally, we show that $\mathcal{M} \subseteq CS(c)$ if $c$ is sufficiently high. We fix $c \geq \overline{c}$. Then, no blocking operation exists such that $y \neq \emptyset$. By contraposition, if $x \notin CS(c)$, then there exists a blocking operation $(x,C,\emptyset,z,0)$. This means that $x \cup z \in X$, and hence $x \notin \mathcal{M}$.
\end{proof}

As a corollary of Proposition \ref{prop:swco1}, we observe that the set of coalitionally stable states is always non-empty if switching costs are sufficiently high.

We now introduce an adjusted unperturbed dynamics which is based on coalitional stability, and we refer to it as \emph{coalition-wise dynamics}. Basically, everything is as in the myopic team-wise dynamics used in Section \ref{section:ss}, with the following difference. At each time two sets of projects, $y$ and $z$, are randomly selected from the set of all subsets of $P$.\footnote{The details of this probabilistic selection are not important for the results as long as every pair $y$, $z$ is chosen with positive probability.} The state of the system changes from the current state $x$ to $(x \setminus y) \cup z$ if there exists a coalition $C$ that has the power to destroy all projects in $y$ and to form all projects in $z$, and has advantage in doing so, i.e., each of its members strictly increases her utility if she behaves in such a manner. Clearly, this coalition is exactly of the type which is excluded by the definition of coalitional stability.

In the coalition-wise dynamics, we denote by $\mathcal{R}(c)$ and $\mathcal{A}(c)$, respectively, the set of recurrent states and the set of absorbing states when the level of switching costs is $c$. A characterization analogous to that in Proposition \ref{ARMB} can be provided for the current setup.

\begin{proposition}\label{prop:swco2}
Take a team formation model satisfying Assumption \ref{ass:v0} and a coalition-wise dynamics. If $c$ is high enough, then $\mathcal{A}(c) = \mathcal{R}(c) = \mathcal{M} = MTS$.
\end{proposition}
 
\begin{proof}
We fix $c \geq \overline{c}$. We already know by definition that $\mathcal{A}(c) \subseteq \mathcal{R}(c)$. 

Suppose now that $x \notin \mathcal{M}$. Then, there exists a project $p \in P$ such that $x \cup \{p\} \in X$. Such a project can be selected in the coalition-wise dynamics and, given Assumption \ref{ass:v0}, it will be formed by agents belonging to $n(p)$. Since no project will ever cease to exist, due to $c \geq \overline{c}$, we can conclude that $x \notin \mathcal{R}(c)$. Hence, $\mathcal{R}(c) \subseteq \mathcal{M}$. 

Moreover, $\mathcal{M} \subseteq \mathcal{A}$, since by starting from a state $x \in \mathcal{M}$ no new project can be formed and existing projects are too costly to be destroyed. 

Finally, we observe that $\mathcal{M} = MTS$ by Lemma \ref{lem:myopic=maximal}.
\end{proof}

To provide results on stochastic stability, we have to introduce a perturbed dynamics. In particular, we can adopt each of the two perturbation schemes of Section \ref{section:ss}. More precisely, in the uniform perturbation scheme, at every time $s$ each project $p \in P$ is hit by an error with an i.i.d.~probability $\epsilon$: if $p$ is an existing project then it disappears, while if $p$ is a non-existing project than it is formed unless $x^s \cup \{p\} \notin X$. In the uniform destructive perturbation scheme, on the other hand, only existing projects can be hit by perturbations. In the coalition-wise dynamics with uniform (destructive) perturbation scheme, we denote by $SS(c)$ the set of stochastically stable states.

\begin{proposition}\label{prop:swco3}
Take a team formation model satisfying Assumption \ref{ass:v0}, a coalition-wise dynamics and either a uniform destructive perturbation scheme or a uniform perturbation scheme. If $c$ is high enough, then $SS(c) = \mathcal{L}$.
\end{proposition}
 
\begin{proof}
We fix $c \geq \overline{c}$. By Propositions \ref{lem:myopic=maximal} and \ref{prop:swco2}, we know that the set of absorbing states is the same as in the model of Section \ref{section:ss}. 

We now consider resistances. Since switching costs are so high that no agent will ever exit from an existing project, we have that a perturbation is required for every project to be destroyed. New projects will instead be formed in the unperturbed dynamics, since they are certainly advantageous due to Assumption \ref{ass:v0}. Therefore, $r^*(x, x') = \ell(x) - \ell(x \cap x')$. Hence, even the resistances between absorbing states are the same as in the model of Section \ref{section:ss}.

Having the same set of absorbing states and the same resistances between them, we can invoke Proposition \ref{prop:stst} to conclude that $SS(c) = \mathcal{L}$. 
\end{proof}

\section{Farsightedly stable sets of states}
\label{App:FSS}
%[FOR ONLINE PUBLICATION]\\
The set of coalitionally stable states, as defined in Section \ref{section:cs}, may be empty.
In this appendix we apply results from  \citet{HMV09,HMV10} to show that it is possible, on the other hand, to define a wider set of \emph{farsightedly stable states}, whose existence is always granted.
The two papers are applied respectively to network formation games and cooperative games, and as our setup generalizes both, the two papers also offer the possibility of finding examples of the non-existence of CS states for our model.

First of all we need to define \emph{improving paths}, with a simple rephrasing of Definition \ref{def:cs}.

\begin{definition}[Improving paths]
\label{def:improving_path}
Given a state $x$ there is an improving path from $x$ to another state $x'$ if there exists a subset $C \subseteq N$ such that
\begin{itemize}
%\item[(i)] there exists a state $y \subseteq x$ such that $\forall ~ p  \in y$, $\exists ~ i \in C$ such that $i \in n(p)$;
%\item[(ii)] there exists a state $z$ such that $\forall ~  p  \in z$, if $j \in n(p)$ then $j \in C$;
\item[(i)] there exists a set of projects $y \subseteq x$ such that $\forall ~ p  \in y$, $\exists ~ i \in C$ such that $i \in n(p)$;
%\item[(ii)] there exists a state $z$ such that $\forall ~  p  \in z$, if $j \in n(p)$ then $j \in C$;
\item[(ii)] there exists a set of projects $z$ such that $z \cap x = \emptyset$ and, $\forall ~  p  \in z$, if $j \in n(p)$ then $j \in C$;
\item[(iii)] $ x'= (x \backslash y ) \cup z \in X$ and for any agent $i \in C$ we have that $u_i ( x' ) > u_i (x)$. 
\eproof
\end{itemize}
\end{definition}

Definition \ref{def:improving_path} generalizes \emph{farsighted improving paths} from Definition 3 in \citet{HMV09} and Definition 1 in \citet{HMV10}. %It is clear that if a state $x$ has no improving paths, then it is CS.
For a state $x$ we can define as $F(x)$ the set of all the states $x'$ that can be reached from $x$ along improving paths, and as $C_{x \rightarrow x'}$ the coalition that is profitably moving from state $x$ to state $x'$.
It is clear that a state $x$ is CS if and only if $F(x) = \emptyset$.

Now, the improving path promoted by a coalition is moving on the partially ordered set $X$ of all possible states to a new state.
What if another coalition were to start a new path from this new state?
This could harm the members of the original coalition, who would not have been \emph{farsighted} enough.
To increase the rationality of the agents the following definition is used.

\begin{definition}
[Farsightedly stable sets of states]
\label{def:consistent_set}
A set of states $S \subseteq X$ is \emph{farsightedly stable} if:
%\begin{enumerate}[(i)]
%\item \label{cond_i} for any $x \in S$, $x' \not \in S$, such that $x' \in F(x)$, there is an $x'' \in F(x')$ such that there is an agent $i \in C_{x \rightarrow x'}$, for which $u_i (x'') < u_i(x)$;
%\item \label{cond_ii} for any $x' \in X \backslash S$, $F(x') \cap S \neq \emptyset$;
%\item \label{cond_iii} there is no $S' \subsetneq S$ such that $S'$ satisfies conditions \ref{cond_i} and \ref{cond_ii} above.
%\end{enumerate}
\begin{itemize}
\item[(i)] for any $x \in S$, $x' \not \in S$, such that $x' \in F(x)$, there is an $x'' \in F(x')$ such that there is an agent $i \in C_{x \rightarrow x'}$, for which $u_i (x'') < u_i(x)$;
\item[(ii)] for any $x' \in X \backslash S$, $F(x') \cap S \neq \emptyset$;
\item[(iii)] there is no $S' \subsetneq S$ such that $S'$ satisfies conditions $(i)$ and $(ii)$ above.
\end{itemize}
\end{definition}

Definition \ref{def:consistent_set} generalizes \emph{farsightedly stable sets} from Definition 4 in \citet{HMV09} and Definition 4 in \citet{HMV10}.
Condition $(i)$ says that if there is an improving path from a state $x$ belonging to $S$, to another state $x'$ outside it, then this is due to the fact that the improving path in question could possibly harm one of the members of the coalition in $C_{x \rightarrow x'}$, because of a new improving path, promoted by another coalition, from $x'$.
Condition $(ii)$ says that from every state outside $X$ there is an improving path into $X$. %there are no improving paths from outside $S$ into $S$.
Finally, since the previous two conditions are trivially satisfied by $X$ itself, condition $(iii)$ says that a set of states $S$ is farsightedly stable if no proper subset of states of $S$ is also farsightedly stable.
A farsightedly stable set $S$ of states always exists, and this is proven by the following proposition, which generalizes Theorem 2 in \citet{HMV09} and Proposition 1 in \citet{HMV10}.

\begin{proposition}[Existence of a farsightedly stable sets of states]
For every team formation model characterized by a quintuple $(N,\vec{w},P,X,\vec{u})$, there exists a farsightedly stable sets of states $S \subseteq X$.
\end{proposition}

\begin{proof}
We proceed by contradiction.
The set $X$ of states, whose cardinality is given by a finite natural number $\omega_0$, satisfies conditions $(i)$ and $(ii)$ in Definition \ref{def:consistent_set}.
Therefore, $X$ may not be  farsightedly stable because it does not satisfy condition $(iii)$, i.e., there is a set $X_1 \subsetneq X$, whose cardinality is given by some finite natural number $\omega_1 < \omega_0$, which also satisfies those two conditions.
If $X_1$ does also not satisfy condition $(iii)$, we can iterate the reasoning.
If we never find a set that satisfies all conditions from Definition \ref{def:consistent_set}, it means that we have found an infinite series $\{\omega_i\}_{i \in \mathbb{N}}$ such that $\omega_{i+1} < \omega_i$ for every $i \in \mathbb{N}$.
Such an infinite series is impossible, and this proves the statement.
\end{proof}

\end{document}